%% file: main.tex
\documentclass[a4paper,UKenglish,cleveref, autoref, thm-restate, numberwithinsect]{lipics-v2021}

\input{header}

\input{authors}

\date{}

\begin{document}

\maketitle
\input{abstract}

\clearpage
\pagenumbering{arabic}

\input{intro}

\input{prelim}

\input{results}

\begin{toappendix}
    \input{examples}
\end{toappendix}

\begin{toappendix}
    \input{stable}

\end{toappendix}

\input{conclusion}

\bibliographystyle{abbrv}
\bibliography{ref}

\end{document}

%% file: header.tex
\nolinenumbers

\usepackage[full]{optional}

\opt{sub}{\usepackage[appendix=append]{apxproof}}
\opt{full}{\usepackage[appendix=inline, repeqn=same]{apxproof}}
\opt{final}{\usepackage[appendix=append]{apxproof}}

\usepackage[disable]{todonotes} 

\setuptodonotes{fancyline}
\usepackage{amsmath}
\usepackage{amsthm}

\newtheoremrep{lemma}[theorem]{Lemma}

\usepackage{graphicx} 
\usepackage{subcaption}

\usepackage{cite}

\usepackage{cleveref}
\usepackage{cancel}
\usepackage[font={small}]{caption}

\usepackage{ amssymb }

\def\longrightharpoonup{\relbar\joinrel\rightharpoonup}
\def\longleftharpoondown{\leftharpoondown\joinrel\relbar}
\def\longrightleftharpoons{\mathop{\vcenter{\hbox{\ooalign{\raise1pt\hbox{$\longrightharpoonup\joinrel$}\crcr\lower1pt\hbox{$\longleftharpoondown\joinrel$}}}}}}
\def\rxn{\mathop{\rightarrow}\limits}  

\def\revrxn{\mathop{\rightleftharpoons}\limits}

\renewcommand{\vec}[1]{\mathbf{#1}}
\newcommand{\vx}{\vec{x}}

\newcommand{\vc}{\vec{c}}
\newcommand{\vd}{\vec{d}}

\newcommand{\vr}{\vec{r}}

\newcommand{\vo}{\vec{o}}
\newcommand{\vM}{\vec{M}}
\newcommand{\vA}{\vec{A}}

\newcommand{\vp}{\vec{p}}
\newcommand{\vi}{\vec{i}}
\newcommand{\vy}{\vec{y}}

\newcommand{\vz}{\vec{z}}

\newcommand{\vu}{\vec{u}}

\newcommand{\vb}{\vec{b}}

\newcommand{\R}{\mathbb{R}}
\newcommand{\Q}{\mathbb{Q}}
\newcommand{\Z}{\mathbb{Z}}
\newcommand{\N}{\mathbb{N}}
\newcommand{\pars}[1]{\left(#1\right)}
\newcommand{\Rp}{\R_{\geq 0}}

\newcommand{\vtau}{{\boldsymbol\tau}}

\newcommand{\ky}{k_{\mathrm{y}}}
\newcommand{\kn}{k_{\mathrm{n}}}
\newcommand{\Vy}{V_\mathrm{y}}
\newcommand{\Vn}{V_\mathrm{n}}

\def\longrightharpoonup{\relbar\joinrel\rightharpoonup}
\def\longleftharpoondown{\leftharpoondown\joinrel\relbar}
\def\longrightleftharpoons{\mathop{\vcenter{\hbox{\ooalign{\raise1pt\hbox{$\longrightharpoonup\joinrel$}\crcr\lower1pt\hbox{$\longleftharpoondown\joinrel$}}}}}}
\def\rxn{\mathop{\rightarrow}\limits}  

\def\revrxn{\mathop{\rightleftharpoons}\limits}

\newcommand{\calC}{\mathcal{C}}
\newcommand{\calD}{\mathcal{D}}
\newcommand{\segto}{\rightsquigarrow}  
\newcommand{\slto}{\to^1}      
\renewcommand{\epsilon}{\varepsilon}
\newcommand{\domR}{\R_{\ge {0}}} 
\newcommand{\yesVotes}{\Upsilon_{\mathrm{yes}}}
\newcommand{\noVotes}{\Upsilon_{\mathrm{no}}}

\newcommand{\vyy}{\ensuremath{V_\mathrm{yy}}}
\newcommand{\vyn}{\ensuremath{V_\mathrm{yn}}}
\newcommand{\vny}{\ensuremath{V_\mathrm{ny}}}
\newcommand{\vnn}{\ensuremath{V_\mathrm{nn}}}

\newcommand{\vm}{\ensuremath{V_\mathrm{m}}}

\newcommand{\Sy}{\ensuremath{Y}}

\usepackage{textcomp} 
\usepackage{listings}
\definecolor{codegreen}{rgb}{0,0.6,0}
\definecolor{codegray}{rgb}{0.5,0.5,0.5}
\definecolor{codepurple}{rgb}{0.58,0,0.82}
\definecolor{backcolour}{rgb}{0.95,0.95,0.92}
\crefname{corollary}{corollary}{corollaries}

%% file: authors.tex
\title{Robust predicate and function computation in continuous chemical reaction networks}

\author{Kim {Calabrese}}{University of California--Davis, Davis, CA, USA \and \url{TO DO}}{Kim Calabrese <ebcalabrese@ucdavis.edu>}
{ORCID TO DO}
{NSF awards 2211793, 1900931.}

\author{David {Doty}}
{University of California--Davis, Davis, CA, USA \and \url{https://web.cs.ucdavis.edu/~doty/}}
{doty@ucdavis.edu}
{https://orcid.org/0000-0002-3922-172X}
{NSF awards 2211793, 1900931 and DoE award DE-SC0024467.}

\author{Mina {Latifi}}{University of California--Davis, Davis, CA, USA \and \url{https://www.linkedin.com/in/mina-latifi/}}{milatifi@ucdavis.edu}
{https://orcid.org/0009-0002-0116-0519}
{NSF awards 2211793, 1900931 and DoE award DE-SC0024467.}

\authorrunning{K. Calabrese, D. Doty, and M. Latifi} 

\ccsdesc[100]{Theory of Computing~Models of Computation} 

\keywords{chemical reaction networks, analog computation, mean-field limit}

%% file: abstract.tex
\begin{abstract}
We initiate the study of ``rate-constant-independent'' computation of Boolean predicates (decision problems) 
and numerical functions
in the continuous model of chemical reaction networks (CRNs),
which model the amount of a chemical species as a nonnegative, real-valued \emph{concentration}, representing an average count per unit volume.
Real-valued numerical functions have previously been studied~\cite{chen2023rate},
finding that exactly the continuous, piecewise rational linear (meaning linear with rational slopes) functions $f: \R_{> 0}^k \to \R_{> 0}$ can be computed \emph{stably} (a.k.a., \emph{rate-independently}),
meaning roughly that the CRN gets the answer correct no matter the rate at which reactions occur.
For example the reactions $X_1 \to Y$ and $X_2+Y \to \emptyset$, 
starting with inputs $X_1 \ge X_2$, 
converge to output $Y$ having concentration equal to the initial difference of inputs $X_1 - X_2$, no matter the relative rate at which each reaction proceeds.

We first show that, contrary to the case of real-valued functions,
continuous CRNs are severely limited in the Boolean predicates they can stably decide, reporting a yes/no answer based only on which inputs are 0 or positive, but not on the exact positive value of any input.

This limitation motivates a slightly relaxed notion of rate-independent computation in CRNs that we call \emph{robust computation}.
The standard mass-action rate model is used,
in which each reaction (e.g., 
$A+B \mathop{\rightarrow}\limits^k C$)
is assigned a \emph{rate} ($A \cdot B \cdot k$ in this example) equal to the product of its reactant concentrations and its \emph{rate constant} $k$.
We say the computation is correct in this model if it converges to the correct output for \emph{any} positive choice of rate constants.
This adversary is weaker than the adversary defining stable computation, the latter being able to run reactions at rates that are not those of mass-action for any choice of rate constants 
(e.g., the stable adversary may deactivate a reaction temporarily, even if all reactants are positive).

We show that CRNs can robustly decide every predicate that is a finite Boolean combination of \emph{threshold predicates}, where a threshold predicate is defined by taking a rational weighted sum of the inputs $\vec{x} \in \R^k_{\ge 0}$ and comparing to a constant, answering the question ``Is $\sum_{i=1}^k w_i \cdot \vec{x}(i) > h$?'', for rational weights $w_i$ and real threshold $h$.
Turning to function computation, we show that CRNs can robustly compute 
any piecewise affine function with rational coefficients,
where threshold predicates determine which affine piece to evaluate for a given input $\vec{x}$.
\end{abstract}

%% file: intro.tex
\section{Introduction}
\label{sec:intro}

A \emph{chemical reaction network} (CRN) is a model of reactions among abstract chemical species, such as the reaction $X+Y \rxn Z$, which indicates that if an $X$ and a $Y$ molecule collide, they can stick together to form a dimer called $Z$.
Since the 19th century~\cite{guldberg1864studies,waage1986studies}, this model has been used to describe and predict the behavior of naturally occurring chemicals.
It was not until the 21st century, however, that the model was repurposed as a \emph{programming} language~\cite{soloveichik2008computation} for describing the desired behavior of synthetically engineering chemicals such as DNA strand displacement systems, both in theory~\cite{soloveichik2010dna} and practice~\cite{chen2013programmable, srinivas2017enzyme}.

CRNs are related to a distributed computing model known as \emph{population protocols}~\cite{AngluinADFP2006},
describing anonymous finite-state agents (molecules) that interact (react) asynchronously in pairs (bimolecular reactions), changing state (chemical species) in response.
Formally, a population protocol is a CRN in which each reaction has exactly two reactants (inputs) and two products (outputs), with unit \emph{rate constants}
(constant multiplier on a reaction rate; see below for definitions.)
Although population protocols are a special case of CRNs, nevertheless many computations achievable with more general CRNs can be simulated exactly by population protocols,
and conversely most impossibility results on population protocols 
also apply to the more general CRN model.

\emph{Discrete} (a.k.a., \emph{stochastic}) CRNs
model the amount of each species as a nonnegative integer representing its exact molecular count~\cite{gillespie1977exact}.
When the number of molecules is very large (a typical DNA nanotechnology experiment may involve over a trillion molecules in a 50 $\mu$L test tube),
CRNs are well-approximated by what is sometimes called the \emph{mean-field limit}, 
the \emph{continuous mass-action} model,
and very commonly
``\emph{the deterministic model}'' 
to contrast with the stochastic nature of the discrete model.
Rather, in the continuous model,
each species has a 
nonnegative, real-valued \emph{concentration} indicating its average count per unit volume.\footnote{
    This continuous mean-field limit has been explored within population protocols as well~\cite{bournez2009on,aupy2011on,bournez2012computing}.
}
In this setting a reaction such as $A+B \rxn^k 2C$ indicates that some real-valued amount of reactants $A$ and $B$ are converted into product species $C$, 
e.g., the reaction could turn $1/3$ units of $A$, 
and the same amount of $B$,
into $2/3$ units of $C$. 
The positive real number $k$ is called a \emph{rate constant};
each reaction has a \emph{rate} equal to the product of the reactant concentrations and $k$: rate $k \cdot A \cdot B$ in this example.
The rates define a system of polynomial ordinary differential equations (ODEs) whose unique solution trajectory\footnote{
    Polynomials are locally Lipschitz, so by the Picard-Lindel\"{o}f Theorem, the ODEs have a \emph{unique} solution; 
    hence the term ``the deterministic model''.
}
indicates how concentrations change over time,
with reactions contributing a positive term to $dA/dt$
\todo{ML: change this and all other instants of $dA/dt$ to $A'$.}
whenever $A$ appears as a product and a negative term when $A$ is a reactant;
see \Cref{sec:prelim:robust} for a formal definition.

\subsection{Computation with CRNs: Related work}
The discrete, integer-valued CRN model has been far more extensively studied for its computational abilities than the continuous model,
so we briefly review what is known for the discrete model.
The two most studied types of computation are Boolean-valued predicates $\phi:\N^k \to \{0,1\}$~\cite{angluin2006stably, angluin2007computational} 
and numeric functions $f:\N^k \to \N$~\cite{CheDotSolNaCo,DotHajLDCRNNaCo}.
We designate special \emph{input species} $X_1,\dots,X_k$ whose initial counts represent the inputs to $\phi$ or $f$.
These may be the only species present initially (``leaderless'' CRNs~\cite{AngluinADFP2006,angluin2007computational,DotHajLDCRNNaCo}), or there may be additional species present initially with counts independent of the inputs (``leader-driven''~\cite{angluin2008fast,CheDotSolNaCo}).
For a function $f:\N^k \to \N$, the count of a special ``output'' species $Y$ represents the output of $f$.
For a predicate $\phi:\N^k \to \{0,1\}$, some species are designated as ``yes'' voters and some as ``no'' voters;
the CRN's output is undefined if voters for both outputs (or neither) are present,
otherwise the output is yes or no depending on which type of voter is present.

Much study has been devoted to \emph{stable} (a.k.a., \emph{rate-independent}) computation,
which means intuitively that the CRN generates the correct output no matter in what order reactions proceed.
In the discrete stochastic model, rates influence the probability of choosing among several competing reactions to be the next reaction to occur,
so rate-independent computation means essentially that,
despite the inherent stochasticity of the model,
nevertheless the correct output is generated with probability 1.\footnote{
    There are caveats to this, unnecessary to understand this paper; see~\cite{p1ccrnJournal} for a thorough discussion.
}
Experimentally, 
rate constants are often difficult to control precisely,
sometimes with a 10-fold difference in rate constants intended to be equal~\cite{srinivas2017enzyme},
whereas reaction stoichiometry (i.e., how many of each reactant is consumed, and how many of each product is produced) is naturally digital and simpler to engineer exactly.
Furthermore there may be violations of the assumptions justifying typical rate laws
(for example the solution may not be well-mixed)
leading to actual reaction rates deviating from predictions.
Whether leaderless or not, exactly the \emph{semilinear} predicates~\cite{angluin2006stably} and functions~\cite{CheDotSolNaCo,DotHajLDCRNNaCo} can be stably computed by discrete CRNs.
See~\cite{angluin2006stably, angluin2007computational,CheDotSolNaCo,DotHajLDCRNNaCo} for formal definitions of these concepts, which are not required to understand this paper.

Studying the (naturally deterministic) continuous model, 
we focus on the ``rate-independent'' characterization of stable computation rather than on probability.
Intuitively we want to capture the idea that the CRN generates the correct output ``no matter the reaction rates''.
In contrast, if rate constants
can be controlled, then continuous CRNs are Turing universal~\cite{fages2017strong}, a consequence of the surprising computational power of polynomial ODEs~\cite{bournez2017polynomial}.

Prior work has completely characterized the real-valued functions $f:\Rp^k \to \Rp$ stably computable by continuous CRNs~\cite{chen2023rate}.
This definition of stable computation is based on a very general notion of reachability, defined formally in \Cref{sec:prelim:stable}.
A function $f$ is stably computable by a continuous CRN if and only if it is piecewise rational linear (meaning a finite union of functions that are linear with rational coefficients, e.g., $\frac{4}{3}x_1 - 5x_2$),
and \emph{positive-continuous},
meaning intuitively that discontinuities can occur only when some input $x_i$ goes from 0 to positive 
(\Cref{defn:positive-continious});
for example the function $f(x_1,x_2) = x_1+x_2$ if $x_2 > 0$, and $f(x_1,x_2) = 3 x_1$ if $x_2=0$,
is positive-continuous but not continuous.

\subsection{Our results}
\label{sec:intro-our-results}
An open question from~\cite{chen2023rate} concerns stable computation of decision problems, a.k.a., Boolean predicates $\phi: \Rp^k \to \{0,1\}$.
In contrast to functions, we show that the stably decidable predicates are severely limited: 
exactly the \emph{detection} predicates can be stably decided (\Cref{thm:stable-detection-predicates}),
those that depend only on whether certain inputs are 0 or positive (\Cref{def:detection-predicate}), 
but not on their exact positive value.
This limitation prompts us to relax the notion of stable computation as suggested by another open question in~\cite{chen2023rate}.
We call the concept \emph{robust computation},
which uses mass-action rate laws
(see \Cref{sec:prelim:robust}),
limited compared to the adversarial rates used in stable computation,
but requires the CRN to converge to the correct output for \emph{every} choice of positive rate constants
(\Cref{def:robustly-decide,def:robustly-compute}).

Our first main result 
(\Cref{thm:multi-threshold-robustly-decidable})
shows that every finite Boolean combination of \emph{threshold predicates} (\Cref{def:threshold-predicate})
is robustly decidable, 
where a \emph{threshold predicate} $\phi(\vx)=1$ if and only if $\sum_{i=1}^k w_i \cdot \vx(i) > h$
for rational constants $w_i \in \Q$ and real constant $h \in \R$.
Intuitively, there are a finite number of hyperplanes,
which cut $\Rp^k$ into a finite number of regions,
and $\phi(\vx)$ depends only on which region $\vx$ is in.

Our second main result 
(\Cref{thm:piecewise-affine-are-robustly-computable})
shows that every \emph{threshold-piecewise rational floor-affine function}
(\Cref{def:piecewise-affine})
can be stably computed,
which are finite unions of ``floor-affine'' components $f_i:\Rp^k \to \Rp$,
where 
a threshold predicate $\phi(\vx)$ indicates whether $f_i$ is the correct affine function to use for input $\vx$ (see \Cref{sec:robust_function_computation}). 
\emph{Floor}-affine means negative outputs are replaced with 0
(\Cref{def:rational-affine}, 
necessary since concentrations are nonnegative),
i.e., $f_i(\vx) = \max(0, g(\vx))$ for some affine $g: \Rp^k \to \R$.
Compared to stable function computation~\cite{chen2023rate},
such functions can be discontinuous even within the strictly positive orthant;
e.g., 
$f(x_1,x_2) = x_1+x_2+2$ if $x_1 < x_2$ and $\max(0, x_1/3 - x_2)$
if $x_1 \ge x_2$.

Threshold predicates and affine functions allow a constant ``offset'', unlike linear functions studied in~\cite{chen2023rate}.
This is because unlike~\cite{chen2023rate},
we allow a notion similar to ``leaders'' in discrete distributed computing models,
called \emph{initial context}:
some non-input species may be present initially, but their initial concentrations do not depend on the input values.
Such initial context can be used to implement these constant offsets.\footnote{
    As noted in \cite[Section 6.2]{chen2023rate},
    allowing initial context in stable computation leads to replacing ``linear'' with ``affine'' in the characterization.
}
We believe our constructions extend naturally to the leaderless setting,
and would limit threshold predicates to comparing to $h=0$ instead of arbitrary constant $h \in \R$,
and limit functions to piecewise \emph{linear} 
(rather than affine),
but we have not explored this in detail.

The high-level goals of stable and robust computation are the same,
which is to formalize a notion of computation in an adversarial environment where reaction rates may deviate from standard models.
However,
robust computation,
being based on ODEs in mass-action kinetics,
requires vastly different (and typically more sophisticated) techniques to reason about.
One goal of this paper is to begin establishing general techniques to prove correctness of such systems.
For example,
\Cref{lem:linear-odes,lem:linear-odes-zero} in 
\Cref{sec:robust_predicate_computation}
are general lemmas that we use repeatedly to reason ``modularly'' when the output of an ``upstream'' CRN $U$ influences a ``downstream'' CRN $D$ in one direction only,
i.e.,
$D$ does not influence $U$.




%% file: prelim.tex
\section{Preliminaries}
\label{sec:prelim}

Let $\N$ denote the nonnegative integers, $\Q$ the rationals, and $\R$ the reals.
For any set $A \subseteq \R$, 
$A_{\ge 0}=A \cap [0,\infty)$
and
$A_{> 0}=A \cap (0,\infty)$.
Given a finite set $F$ and a set $S$,
let $S^F$ denote the set of functions $\vc: F \to S$.
In the case of $S = \R$ (resp., $\N$), we view $\vc$ equivalently as a real-valued (resp., integer-valued) vector indexed by elements of $F$.
Given $a \in F$, we write $\vc(a)$, to denote the real number indexed by $a$.
The notation $\Rp^F$ is defined similarly for nonnegative real vectors.
We view such vectors equivalent as multisets,
e.g., the multiset $\vc = \{A,A,B,C,C,C\}$ where 
$\vc(A)=2, \vc(B)=1, \vc(C)=3, \vc(x)=0$ for all $x \in F \setminus \{A,B,C\}$.
For a function of time $A:\Rp \to \R$ and $c \in \R \cup \{\infty\}$, we write $A \to c$ to denote $\lim_{t\to\infty} A(t) = c$; 
most frequently, $A$ will be a chemical species concentration, or some function thereof.

\begin{definition}
\label{def:rational-affine}
$f:\R^k \to \R$ is \emph{rational affine} if for some rational $w_1,\dots,w_k \in \Q$ and real $h \in \R$,
for all $\vx \in \R^k$,
$f(\vx) = h+\sum_{i=1}^k w_i \cdot \vx(i)$.\footnote{
    It would seem more natural to require $h$ also to be rational, 
    or conversely to allow each $w_i$ to be real, but this distinction is relevant with CRNs:
    $h$ will come from the initial concentration of some species, or its negation, whereas the $w_i$'s will come from ratios of integer stoichiometric reaction coefficients.
    Similar reasoning applies to \Cref{def:threshold-predicate}.
}
If $h=0$ we say $f$ is \emph{rational linear}.
$f$ is \emph{rational floor-affine} if $f(\vx) = \max(0,g(\vx))$ for some rational affine $g$.\footnote{
    The significance of rational floor-affine functions is that CRNs can only output nonnegative concentrations;
    reactions such as $X_1 \rxn Y$ and $X_2 + Y \to \emptyset$ technically compute, not the affine function $g(x_1,x_2) = x_1-x_2$,
    but the floor-affine function $f(x_1,x_2) = \max(0, g(x_1,x_2))$. 
}
\end{definition}

The next two definitions capture the class of predicates and functions we show are robustly computable by continuous CRNs in our main results, \Cref{thm:multi-threshold-robustly-decidable,thm:piecewise-affine-are-robustly-computable}.

\begin{definition}
\label{def:threshold-predicate}
    $\phi : \R^k \to \{0,1\}$ is a \emph{threshold predicate}
    with rational \emph{weights} $w_1,\ldots,w_k \in \Q$ and real \emph{threshold} $h \in \R$ if, for all $\vx \in \Rp^k$,
    $\phi(\vx) = 1 \iff \sum_{i=1}^k w_i\cdot \vx(i) > h$.
    We say $\phi$ is a \emph{multi-threshold predicate} if it is a finite Boolean combination of threshold predicates.
\end{definition}

\begin{definition}
\label{def:piecewise-affine}
    $f:\Rp^k \to \Rp$ is \emph{threshold-piecewise rational floor-affine}
    if there is a finite set of rational floor-affine functions 
    $f_1,\dots,f_l: \Rp^k \to \Rp$,
    known as the \emph{affine components} of $f$,
    and multi-threshold predicates 
    $\phi_1,\dots,\phi_l:\Rp^k \to \{0,1\}$ such that
    \begin{enumerate}
    \item 
        The sets 
        $\phi_1^{-1}(1),\dots,\phi_l^{-1}(1)$ are a partition of $\Rp^k$,
        i.e., for each $\vx \in \Rp^k$,
        $\phi_i(\vx)=1$
        for exactly one $1 \leq i \leq l$.
        
    \item 
        For each $i$ and $\vx \in \Rp^k$,
        if $\phi_i(\vx)=1$ then $f(\vx) = f_i(\vx)$,
        i.e., $\phi_i(\vx)$ indicates whether $f_i$ is the correct affine component defining $f(\vx)$.
    \end{enumerate}
\end{definition}

\subsection{Chemical reaction networks}

Throughout this paper, let $\Lambda$ be a finite set of chemical \emph{species}.
Given $S\in \Lambda$ and \emph{state} $\vc \in \Rp^\Lambda$,
$\vc(S)$ is the \emph{concentration of $S$ in $\vc$}.
For any $\vc\in \Rp^\Lambda$, let $[\vc] = \{S \in \Lambda \ |\ \vc(S) > 0 \}$, the set of species \emph{present} in $\vc$
(a.k.a., the \emph{support} of $\vc$).
We write $\vc \leq \vc'$ to denote that $\vc(S) \leq \vc'(S)$ for all $S \in \Lambda$.
Given $\vc,\vc' \in \Rp^\Lambda$, we define the vector component-wise operations of addition $\vc+\vc'$, subtraction $\vc-\vc'$, and scalar multiplication $x \vc$ for $x \in \R$. 

A \emph{reaction} over $\Lambda$ is a triple $\alpha = (\vr,\vp,k) \in \N^\Lambda \times \N^\Lambda \times \R_{> 0}$, 
such that $\vr \neq \vp$,
specifying the stoichiometry of the \emph{reactants} $\vr$ and \emph{products} $\vp$,
and the \emph{rate constant} $k$.
For instance, given $\Lambda=\{A,B,C,D\}$, the reaction $A+2B \rxn^{6.7} A+3C$ is the triple $({(1,2,0,0)},{(1,0,3,0)}, 6.7).$ 

A \emph{chemical reaction network (CRN)} is a pair $\calC=(\Lambda,R)$, where $\Lambda$ is a finite set of chemical \emph{species},
and $R$ is a finite set of reactions over $\Lambda$.
A \emph{state} of a CRN $\calC=(\Lambda,R)$ is a vector $\vc \in \Rp^\Lambda$.
Given a state $\vc$ and reaction $\alpha = (\vr,\vp,k)$, we say that $\alpha$ is \emph{applicable} in $\vc$ if $[\vr] \subseteq [\vc]$ (i.e., $\vc$ contains positive concentration of all of the reactants).
If no reaction is applicable in state $\vc$, we say $\vc$ is \emph{static}. 

The next two definitions are ``syntactic'' preparation for stating how a CRN can compute a predicate or function;
the ``semantic'' definitions of stable 
(\Cref{def:stably-decide,def:stably-compute}) 
and robust 
(\Cref{def:robustly-decide,def:robustly-compute}) 
computation will use these definitions to state under what conditions a CRN ``correctly'' computes.
The first definition is for Boolean predicates.

\begin{definition}
\label{def:crd}
A \emph{chemical reaction decider} (CRD) is a tuple $\mathcal{D} = (\Lambda,R,\Sigma,\yesVotes, \noVotes,\vi)$ where $(\Lambda,R)$ is a CRN, $\Sigma \subseteq \Lambda$ is the set of \emph{input species}, $\yesVotes \subseteq \Lambda$ is the set of \emph{yes voters}, $\noVotes \subseteq \Lambda \backslash \yesVotes$ is the set of \emph{no voters}, and 
$\vi \in \Rp^{\Lambda \setminus \Sigma}$ 
is the \emph{initial context}.
\end{definition}

Intuitively initial context $\vi$ refers to fixed initial concentrations for non-input species, independent of the input value.
A CRD's initial state for predicate input $\vx \in \Rp^{k}$ is then $\vx + \vi$,
where we assume some fixed ordering $X_1,\dots,X_k$ of input species to interpret a vector $\vx \in \Rp^k$ as a state $\vx \in \Rp^\Sigma$.
The next definition is used for computing numeric functions, identifying a special species $Y$ whose concentration represents output:

\begin{definition}
\label{def:crc}
A \emph{chemical reaction computer} (\emph{CRC}) is a tuple $\calC = (\Lambda,R,\Sigma, Y,\vi)$ where $(\Lambda,R)$ is a CRN, 
$\Sigma\subseteq \Lambda$ is the set of \emph{input species}, 
and $Y$ is the \emph{output species},
and
$\vi \in \Rp^{\Lambda \setminus \Sigma}$ 
is the \emph{initial context}.
\end{definition}

\subsection{Robust (rate-\emph{constant}-independent) computation}
\label{sec:prelim:robust}

See \Cref{sec:prelim:stable} for a formal definition of stable computation.
Stable computation requires a CRN to work against a very powerful adversary who can essentially set the rate of each reaction at each time arbitrarily.
This means in particular that the CRN works under a variety of rate laws besides mass-action (defined below).
Here we consider a weaker adversary, one that cannot control the rate law---that will be mass-action---but that \emph{can} set the parameters of the rate law, known as \emph{rate constants}.
Crucially, these are constant with respect to \emph{time}:
the adversary can choose arbitrary positive values for these rate constants, but sets them to those values at time $t=0$,
and the rate constants stay at those values for all future $t > 0$.\footnote{
    If we allowed the adversary to change the rate constants over time, then it could mimic the stable computation adversary by adjusting rate constants so as to target particular absolute rates at each time.
}
A CRN robustly computes a function or predicate if it computes the correct output against this adversary: i.e., if the mass-action rate law converges to the correct output, 
no matter which positive rate constants are chosen.


A CRN $\calC = (\Lambda, R)$
under the \emph{mass-action rate law} is governed by a system of polynomial ordinary differential equations (ODEs) with a variable $S(t)$ representing the concentration of  species $S$ at time $t$. 
The \emph{rate} 
$\rho_t(\alpha)$ of a reaction $\alpha = (\vr,\vp,k)$ at time $t$ is 
\todo{DD: oooooops... This said $S(t)^{\vp(S)}$ in the DISC submission. We'll see if anyone notices}
$\rho_t(\alpha)=k \cdot \prod_{S \in \Lambda} S(t)^{\vr(S)}$,
i.e., the rate constant times each reactant concentration at time $t$.
For example, the rate of $A+2B \rxn^{4.5} C$ is 
$4.5 \cdot A(t) \cdot B(t)^2$.
Each reaction $\alpha=(\vr,\vp,k)$ contributes a term $\rho_t(\alpha) \cdot (\vp(S) - \vr(S))$ to the ODE of each species $S$ that is net produced or consumed;
the term is $\alpha$'s rate $\rho_t(\alpha)$ times the net stoichiometry of $S$ in $\alpha$
(positive if $S$ is net produced by $\alpha$, e.g., $S \rxn 3S$ net produces 2 $S$'s since $\vp(S)-\vr(S)=2$
and negative if $S$ is net consumed, e.g., $2S \rxn S$ net consumes 1 $S$,
i.e., $\vp(S)-\vr(S)=-1$). 
For example, the CRN with reactions
     $A + 2B \rxn^{k_1} 3C$ 
    and $2C \rxn^{k_2} C$
corresponds to the ODEs
\begin{align*}
        A'(t) &= -k_1A(t)B(t)^2\\
        B'(t) &= - 2k_1A(t)B(t)^2\\
        C'(t) &= 3k_1A(t)B(t)^2 -k_2C(t)^2
\end{align*}
Given a CRN $\calC = (\Lambda,R)$ let $\vA : \domR^\Lambda \to \domR^R$ map each state $\vd$ of $\calC$ to the vector $\vA(\vd)$ of instantaneous reaction rates in state $\vd$, as given by the mass-action ODEs. With the example reaction above, the state of concentrations $\vd = (1,2,0)$ would be mapped to the flux vector $\vA(\vd) = (-4k_1,-8k_1,12k_1)$.

\todo{DD: We should emphasize that what is in the paragraph below is not a new definition (in particular equation \eqref{eq:mass_action_IVP} is defining the same thing as the mass-action rate law paragraph above), but simply rephrasing the above to give terms to some objects like $\vtau$. We should put thought into eliminating redundancy between these two paragraphs.}
Define the 
$|\Lambda| \times |R|$ \emph{stoichiometry matrix} $\vM$ such that,
for species $S \in \Lambda$ and reaction $\alpha = (\vr,\vp) \in R$,
 $\vM(S,\alpha) = \vp(S) - \vr(S)$
is the net amount of 
$S$ produced by $\alpha$ (negative if $S$ is consumed).
For example, if we have the reactions $X \to Y$ and $X + A \to 2X + 3Y$, and if the three rows correspond to $A$, $X$, and $Y$, in that order, then
$
    \vM =
    \left(
      \begin{array}{cc}
         0 & -1 \\
        -1 &  1 \\
         1 &  3 \\
      \end{array}
    \right).
$
Then the vector $\vM \cdot \vA(\vd)$ gives the rate at which each species concentration is changing in state $\vd$.
Given an 
initial state $\vc$, the \emph{mass-action trajectory} $\vtau : [0,t_{\mathrm{max}})\to \domR^\Lambda$ \emph{starting at $\vc$}, is the solution to the initial value problem 
\begin{equation}
    \label{eq:mass_action_IVP}
    \frac{d\vtau}{dt} = \vM \cdot \vA(\vtau(t)),~\vtau(0) = \vc
\end{equation}
where $t_{\mathrm{max}} \in \R_{\ge0} \cup \{\infty\}$.\footnote{
    Although $t_\mathrm{max} = \infty$ for ``typical'' CRNs, there are pathological CRNs such as $2X \rxn 3X$, which diverge to $\infty$ in finite time;
    for instance with $X(0)=1$ and unit rate constant, this CRN has the solution $X(t) = 1 / (1-t)$, which goes to $\infty$ as $t \to 1$, so $t_\mathrm{max} = 1$ for this CRN.
}
That is, $\vtau(t) = (S_1(t),S_2(t),\ldots,S_{|\Lambda|}(t))$ is the vector whose $i$th component is the concentration of species $S_i \in \Lambda$ at time $t$. While some CRNs induce ODEs with solutions that do not exist for all time, mass-action ODEs are locally Lipschitz, implying that a unique solution to \eqref{eq:mass_action_IVP} always exists on some interval. We now define what it means for CRN to decide a predicate in the mass action model with adversarial rate constants.

\begin{definition}[robustly decide]
\label{def:robustly-decide}
    Let $\phi : \domR^k \to \{0,1\}$ be a predicate. We say a CRD $\calD = (\Lambda,R,\Sigma,\yesVotes,\noVotes,\vi)$ 
    \emph{robustly decides}
    (a.k.a., \emph{rate-constant-independently decides})
    $\phi$ if,
    for any choice of strictly positive rate constants and every $\vx \in \domR$,
    the following holds.
    
    Let $\vtau$ be the mass-action trajectory of $\calC$, 
    starting at state $\vx + \vi$.
    If $\phi(\vx)=1$ (resp. 0), let $\Upsilon_\mathrm{C}=\yesVotes$ (resp. $\noVotes$) 
    be the correct voters, 
    and let $\Upsilon_\mathrm{I} = \noVotes$ (resp. $\yesVotes$) be the incorrect voters.
    Define $C(t) = \sum_{V \in \Upsilon_\mathrm{C}} V(t)$ be the sum of concentrations of correct voters.
    Then 
    $
        \liminf_{t \to \infty} C(t) > 0
    \text{ and }
        \lim_{t \to \infty} I(t) = 0
    \text{  for all }
        I \in \Upsilon_\mathrm{I}.
    $
\end{definition}

In other words, the CRD deciding $\phi(\vx)$ starts with input species concentrations defined by $\vx$,
and other initially present species indicated by $\vi$
(whose initial concentrations are the same for all inputs $\vx$).
The CRD proceeds by mass-action dynamics as defined above, 
and converges to a state with only correct voters present. 
Since concentrations are nonnegative,
requiring each individual incorrect voter to converge to 0 is equivalent to requiring their sum to converge to 0,
whereas we do not require any individual correct voter to stay bounded above 0, only the sum of correct voters,\footnote{
    For example two yes voters could oscillate between 0 and 1, so long as they always sum to at least $0.1$.
}

We now define what it means for a CRC to robustly compute a real-valued function.

\begin{definition}[robustly compute]
\label{def:robustly-compute}
    Let $f:\domR^k \to \Rp$. 
    We say a CRC $\calC = (\Lambda,R,\Sigma,Y,\vi)$ \emph{robustly computes}
    (a.k.a., \emph{rate-constant-independently computes})
    $f$ if,
    for any choice of strictly positive rate constants and every $\vx \in \domR$,
    the component $Y(t)$ of $\calC$'s mass action trajectory
    starting from $\vi + \vx$ satisfies 
    $
    \lim_{t \to \infty}Y(t) = f(\vx).
    $
\end{definition}

The full definition of stable computation is given in \Cref{sec:stable_predicate_computation}, based on a formal definition of reachability in continuous CRNs.
The definition intuitively says a CRN can reach from state $\vx$ to $\vc$ if one can run some reactions starting at $\vx$ and reach to $\vc$, without ever running a reaction when one of its reactants is 0.
We say that a CRC \emph{stably computes} a function $f$ (or stably decides a predicate $\phi$) if, starting from initial state $\vx$ encoding the input,
for any state $\vc$ reachable from $\vx$, there is a ``correct'' state $\vo$ reachable from $\vc$ (correct meaning the output in $\vo$ equals $f(\vx)$ or $\phi(\vx)$),
that is also \emph{stable},
meaning that every state $\vo'$ reachable from $\vo$ has the same output as $\vo$.

The next definition  connects stable and robust computation for some specially structured CRNs.
Intuitively, a CRN is \emph{feedforward} if there is an ordering of species so that every reaction producing a species consumes another species earlier in the ordering; 
formally:

\begin{definition}
\label{def:feedforward}
A CRN $\calC = (\Lambda,R)$ with 
is \emph{feedforward} if $\Lambda$ can be ordered 
$\Lambda = \{S_1,S_2, \dots,S_n\}$ so that,
if for each reaction $\alpha=(\vr,\vp,k) \in R$ and $S_j$ where $\vp(j)>\vr(j)$, there is $S_{i} \in \Lambda$ with $i<j$ such that $\vp(i)<\vr(i)$.
\end{definition}

The following was shown in~\cite[Corollary 4.11]{chen2023rate} (in different but equivalent terms).

\begin{lemma}
\label{lem:feedforward-stably-compute-implies-robustly-compute}
    Each feedforward CRC stably computing function $f$ 
    also
    robustly computes $f$.
\end{lemma}

Essentially the same proof shows the following.

\begin{lemma}
\label{lem:feedforward-stably-decide-implies-robustly-decide}
    Each feedforward CRD stably deciding a predicate $\phi$ also robustly decides $\phi$.
\end{lemma}

%% file: results.tex
\section{Robust computation by continuous CRNs}
\begin{toappendix}
  \label{apx:robust_predicate_computation}
\end{toappendix}

\subsection{Boolean combinations of threshold predicates are robustly decidable}
\label{sec:robust_predicate_computation}







We begin by proving a technical lemma that relates known asymptotic behavior of certain species to the desired asymptotic behavior of others that depend on them.
Intuitively, we think of $p$ and $g$ as functions of concentration whose asymptotic behavior has already been analyzed. In particular, we think of species involved in $p$ and $g$ as belonging to an ``upstream'' CRC (or CRD) $\mathcal{C}_\mathrm{U}$,
whose outputs are used in reactions of a ``downstream'' CRC (or CRD) $\mathcal{C}_\mathrm{D}$ that influence concentration of species $F$.
These species influence the concentration of species $F$, which evolves as $f(t)$. However, $f$ does not affect the concentrations of the species involved in $p$ and $g$. Assuming that the ODE describing $f$ is of the form $f(t)' = g(t) - p(t)f(t)$, we can use \Cref{lem:linear-odes,lem:linear-odes-zero} to reason about asymptotic behavior of $f$.

\opt{sub}{
    Proofs of \Cref{lem:linear-odes,lem:linear-odes-zero} are given in \Cref{apx:robust_predicate_computation}.
}

\begin{lemmarep}
\label{lem:linear-odes}
    Let $p,g: \R_{\ge 0} \to \R$, be differentiable functions, 
    with $p(t) > 0$ and $g(t) \ge 0$ for all $t \in \R_{\ge 0}$. 
    Let $K = \frac{g(0)+p(0)f(0)}{p(0)}$ be a constant. 
    If $f: \R \to \R$ is differentiable and satisfies the first order linear ODE 
    $
    f'(t) = g(t)-p(t)f(t),
    $
    then for all $t \in \Rp$,
    \[
    f(t) \le \frac{g(t)}{p(t)}+K\exp\left(-\int_0^t p(s)ds\right).
    \]
\end{lemmarep}

\begin{proof}
    As $f$ satisfies the first-order linear ODE $f'(t) + p(t)y = g(t)$, by the method of integrating factor, it has a closed-form solution of 
    \[
    f(t) = \frac{1}{\mu(t)}\int_0^t \mu(s)g(s)ds +\frac{f(0)}{\mu(t)}
    \]
    where $\mu(t)$ is the so-called integrating factor, defined as $\mu(t) = \exp(\int_0^t p(s)ds)$. Note that $\mu$ satisfies the differential equation $\mu'(t) = \mu(t)p(t).$ As $p(t) \ne 0$ for all $t$, we can write $\mu(t)$ as $\mu(t) = \frac{\mu'(t)}{p(t)}$. With this we can realize $f(t)$ as 
\begin{equation*}
    f(t) = \frac{1}{\mu(t)}\int_0^t \frac{g(s)}{p(s)}\mu'(s)ds + \frac{f(0)}{\mu(t)}.
\end{equation*}
We now apply integration by parts, which yields
\begin{equation*}
    f(t) = \frac{1}{\mu(t)}\frac{\mu(t)g(t)}{p(t)}- \frac{1}{\mu(t)}\int_0^t\mu(s)\left(\frac{g(s)}{p(s)}\right)'ds+ \frac{f(0)}{\mu(t)}.
\end{equation*}

As $p(t) > 0$, this implies that $\mu(t) \ge 1$. We can use this to estimate $f(t)$ as
\begin{align*}
f(t) &= \frac{g(t)}{p(t)}- \frac{1}{\mu(t)}\int_0^t\mu(s)\left(\frac{g(s)}{p(s)}\right)'ds + \frac{f(0)}{\mu(t)}\\
    &\le \frac{g(t)}{p(t)}- \frac{1}{\mu(t)}\int_0^t\left(\frac{g(s)}{p(s)}\right)'ds+ \frac{f(0)}{\mu(t)}\\
    &= \frac{g(t)}{p(t)}- \frac{1}{\mu(t)}\frac{g(t)}{p(t)}+ \frac{p(0)f(0)+g(0)}{p(0)\mu(t)}\\
    &\le \frac{g(t)}{p(t)}+ \frac{f(0)p(0)+g(0)}{p(0)\mu(t)}
\end{align*}
which shows the result.\qedhere
\end{proof}
This lemma allows us to cleanly demonstrate that the concentration of particular voting species $V$ converges to zero. 
In this context, the function $p$ will represent the concentration of species whose presence causes $V$ to be consumed, and $g$ represents the concentration of species whose presence causes $V$ to be produced. 
A common pattern in our correctness proofs will be to rearrange the mass-action ODEs for an incorrect voting species into a first order linear ODE, apply \cref{lem:linear-odes}, 
and then argue that $g(t)$ converges to zero while $p(t)$ converges to a positive value (or converges to zero slower than $g$). 
Further, the non-negativity condition of the hypothesis is trivially satisfied as components of the mass action trajectory are always non-negative, and for some CRNs, strictly positive with suitable initial conditions. 
\Cref{lem:linear-odes} is not applicable if $p(t)$ is not strictly positive, but in this case we can still find a bound of a similar form. 

\begin{lemmarep}
\label{lem:linear-odes-zero}
    Let $p,g: \R_{\ge 0} \to \R$ be differentiable functions, 
    with $p(t) \ge 0$ and $g(t) \ge 0$ for all $t \in \R_{\ge 0}$.
    Let $K = \frac{g(0)+f(0)(p(0)+1)}{p(0)+1}$ be a constant.
    If $f: \R \to \R$ is differentiable and satisfies the first order linear ODE 
    $
    f'(t) = g(t)-  p(t)f(t),
    $
    then for all $t \in \Rp$,
    \[f(t) \le \frac{5g(t)}{p(t)+(t^2+1)^{-1}}+K\exp\left(-\int_0^t p(s)ds\right).
    \] 
\end{lemmarep}

\begin{proof}
    We again apply the formula for the solution to a first-order linear ODE to write
    \begin{equation}
        f(t)= \frac{1}{\mu(t)}\int_0^t\mu(s)g(s)ds +\frac{f(0)}{\mu(t)}
    \end{equation}
    where again $\mu(t)$ is defined as $\exp(\int_0^t p(s) ds)$. Observing that we can bound $\mu$ by the expression
    \begin{equation}
        \gamma(t) =\exp\left(\int_0^t p(s)ds+\arctan(t)\right)
    \end{equation}
    We obtain that
    \begin{equation}
        f(t) \le\frac{1}{\mu(t)}\int_0^t\gamma(s)g(s)ds + \frac{f(0)}{\mu(t)}.
    \end{equation}
    Observe that $\gamma$ satisfies the differential equation $\gamma'(t) = (p(t)+\frac{d}{dx}\arctan(t))\gamma(t)$. As $\frac{d}{dx}\arctan(t) = \frac{1}{t^2+1}$ is positive for all $t$, we can solve for $\gamma$ in terms of $\gamma'$ and apply the same trick we did in the proof of \cref{lem:linear-odes} to rewrite the integral as
    \begin{equation*}
        \int_0^t\gamma(s)g(s)ds = \int_0^t\gamma'(s) \pars{\frac{g(s)}{p(s)+\frac{1}{s^2+1}}}ds
    \end{equation*}
    We can now apply integration by parts and similar arguments to the proof of \cref{lem:linear-odes} to obtain the desired bound. The factor of 5 comes from the fact that $e^{\arctan(x)} \le e^{\pi/2}\le 5$.
\end{proof}
\begin{toappendix}
    We remark that arctangent was a rather arbitrary choice of function, any bounded, strictly positive function with a strictly positive derivative would suffice for our purposes. 
\end{toappendix}

We begin our positive results by showing the \emph{majority predicate} $\phi: \R^2 \to \{0,1\}$ defined by $\phi(a,b) = 1$ if and only if $a > b$ is robustly decidable. We remark that such a predicate is in general not a detection predicate, demonstrating that the class of robustly decidable predicates is strictly larger than the class of stably decidable ones. 
\begin{lemma}
\label{thm:majority}
The majority predicate $\phi:\domR^2 \to \{0,1\}$ is robustly decidable. 
\end{lemma}








\begin{proof}
We construct the CRD $\calD = (\Lambda,R,\Sigma,\yesVotes,\noVotes,\vi)$ as follows.
Let $\Upsilon_\mathrm{yes} = \{Y\}$ and $\Upsilon_\mathrm{no} = \{ N \}$ be the yes voter and no voter respectively. Let the initial state be $\{ a A, b B, 1 Y, 1 C\}$ where $\{1Y, 1C\}$ is the initial context $\vi$. We add the following reactions:
\begin{eqnarray}
    A + N &\rxn^{k_1}& A + Y
    \label{rxn:maj-an}
    \\
    B + Y &\rxn^{k_2}& B + N
    \label{rxn:maj-by}
    \\
    A + B &\rxn^{k_3}& \emptyset
    \label{rxn:maj-ab}
    \\
    C + Y &\rxn^{k_4}& C + N
    \label{rxn:maj-cy}
    \\
    3C &\rxn^{k_5}& \emptyset
    \label{rxn:maj-3c}
\end{eqnarray}

We consider the following 3 cases separately. 

\begin{description}
\item[$a>b$:]
To argue that we converge to a correct vote, we must show that $\lim_{t \to \infty} Y(t) = 1$. This is equivalent to showing $\lim_{t \to \infty} N(t) = 0$, as $Y(t) + N(t) = 1$ for all $t$. We first observe that the only reactions that change the concentration of species $A$, $B$, and $C$ strictly decrease the concentration of each species.  
Therefore, as $t$ approaches infinity, the limits of their concentrations satisfy $\lim_{t \to \infty} A(t) = a-b$ and $\lim_{t \to \infty} B(t) = \lim_{t \to \infty} C(t) =0$, as the initial concentration of $A$ is strictly greater than the concentration of $B$. \footnote{This can also be seen by application of \cref{lem:feedforward-stably-compute-implies-robustly-compute}. The sub-CRN with reactions $A+B \to \emptyset$ and $3C \to \emptyset$ is feedforward and stably computes $f(a,b,c) = a-b$ when $a > b$, converging to the static state $\{(a-b)A,0B,0C\}$. By \cref{lem:feedforward-stably-compute-implies-robustly-compute}, this restricted CRN robustly computes $a-b$.} 

We observe the CRC induces this mass-action ODE for the concentration of $N$:
$
N'(t) = -k_1A(t)N(t)+k_2B(t)Y(t)+k_3C(t)Y(t).
$
Using the fact that $Y(t) + N(t) = 1$, we may rewrite the ODE in terms of only $N(t)$ and species whose limit is known:
$
N'(t) + N(t)(k_1A(t)+k_2B(t)+k_3C(t)) = k_2B(t) + k_3C(t).
$
We now want to apply \cref{lem:linear-odes}, to connect the asymptotic behavior of $A,B$ and $C$ to that of $N$. To see that this lemma is applicable, we observe that with the initial conditions $A(0) = a>0$, the function $A(t)$ is always positive as it monotonically decreases and approaches a positive limit. This also implies the lower bound $A(t) \ge a-b$. Further, $B(t)$ and $C(t)$ are both non-negative, so we can apply \cref{lem:linear-odes} with $f(t) = N(t)$, $p(t) = k_1A(t)+k_2B(t)+k_3C(t)$ and $g(t) = k_2B(t) + k_3C(t)$ to obtain the bound
\begin{align*}
    N(t) &\le \frac{k_2B(t) + k_3C(t)}{k_1A(t)+k_2B(t)+k_3C(t)} + K\exp\left(-\int_0^t k_1A(s)+k_2B(s)+k_3C(s)ds\right)\\
    &\le \frac{k_2B(t) + k_3C(t)}{k_1A(t)+k_2B(t)+k_3C(t)} + K\exp\left(-(a-b)k_1t\right)\\
\end{align*}
Since $B$ and $C$ converge to 0 as $t$ approaches infinity, this shows $\lim_{t \to \infty} N(t) = 0$.
\item[$a<b$:]
Symmetric to the previous case.
\item[$a=b$:]
We want to show that in this case CRN converges to the static state $\vc = \{1 N\}$. 

With the initial conditions $A(0) = B(0) = a$, it can be found that the functions $A(t)$ and $B(t)$ are equal from their mass action ODEs
$ A'(t) = B'(t) = -A(t)B(t).$   
We can then derive that the closed form for the concentrations $A(t)$ and $B(t)$ is $\frac{a}{k_1t+1}$, where $a$ is the initial concentration of $A$. 
    Similarly, we can compute that $C$ has closed form $C(t) = \frac{1}{\sqrt{2k_3t +1}}$. 
    We consider the mass action ODE for $Y'(t)$ 
    \[
      Y'(t) + \left(\frac{2a}{k_1t+1} + \frac{1}{\sqrt{2k_3t+1}}\right)Y(t) = \frac{a}{k_1t+1}
    \]
    This is a first-order linear ordinary differential equation, so we can find an explicit solution for $Y$ by way of integrating factor:
    \[
    Y(t) = 
    \frac
    {\int_0^t (as+1)^{a/k}e^{\sqrt{2k_3s+1}/k_3} {d}s+1}
    {(at+1)^{2a/k}e^{\sqrt{2k_3t+1}/k_3}}.
    \]

    Taking the limit as $t \to \infty$, observe that the numerator and the denominator both approach $\infty$ for any positive choice of $k_1$ and $k_3$. 
    Thus, 

    \begin{align*}
        \lim_{t \to \infty} Y(t) &= \lim_{t\to\infty}Y'(t)
        &\text{L'H\^{o}pital's rule}
    \\ &= 
        \lim_{t \to \infty} \frac{\frac{a}{k_1t+1}}{\frac{2a}{k_1t+1}+\frac{1}{\sqrt{2k_3t+1}}}
    \\ &=
        \lim_{t \to \infty} \frac{a}{k_1t+1}\times \frac{(k_1t+1)\sqrt{2k_3t+1}}{2a\sqrt{2k_3t+1}+k_1t+1}
    \\ &=
        \lim_{t \to \infty} \frac{O\pars{t^{3/2}}}{O\pars{t^2}} 
    = 0.
    \end{align*}
    Since $Y(t)$ approaches $0$ as $t \to \infty$, $N(t)$ approaches 1 as desired.
    \qedhere
\end{description}
\end{proof}

The main difficulty of computing majority comes from correctly deciding the case of $a = b$. If the inputs $a$ and $b$ were guaranteed to be unequal, then the first three reactions in our construction \eqref{rxn:maj-an}, \eqref{rxn:maj-by}, and \eqref{rxn:maj-ab} would be sufficient to decide the predicate.
These three reactions are no longer sufficient when $a = b$, 
since one can imagine that an adversary could set rate constants such that reaction \eqref{rxn:maj-ab} consumes most of $A$ and $B$, hence making the rate of \eqref{rxn:maj-by} far too slow to change all of the yes voters $Y$ into no voters.\footnote{Indeed, one can find that the closed form for $Y(t)$ for these reactions is given by $Y(t) = 1-(k_1t+1)^{-a/k_1}$. If an adversary sets $k_1 = a$, then $\lim_{t \to \infty} Y(t) = 1$ which is incorrect here.} This is where the addition of the auxiliary species $C$ is necessary. As $C$ converges to zero \textit{asymptotically} slower than $A$, its presence can speed up the production of $N$.
This use of different asymptotic decays is a powerful tool of mass-action, not available for stable computation,
which explains CRDs can robustly decide majority but cannot stably decide it.

See \Cref{fig:majority_lessA,fig:majority_equalAB} for simulation data, demonstrating that the CRD of \Cref{thm:majority} converges with equal rate constants, and even with ``adversarial'' rate constants where reactions that ``fight'' correct convergence have larger rate constants than reactions that encourage correct convergence.

Next, we show that threshold predicates are robustly decidable by CRDs. 

\begin{lemmarep}
\label{thm:threshold-robustly-decidable}
Every threshold predicate is robustly decidable by a continuous CRC.
\end{lemmarep}
\opt{sub}{
\begin{proofsketch}
    A full proof is given in \Cref{apx:robust_predicate_computation}.
    It works by reducing the problem of deciding $\phi$ to that of deciding majority as in \Cref{thm:majority}.
    For example, to decide whether $2x_1 - x_2/3 + \frac{5}{4} x_3 > 4$,
    we start with $4 B$,
    and we convert positive terms to $A$ and negative terms to $B$ with the correct rational multipliers:
    reactions 
    $X_1 \rxn 2A$
    and
    $4 X_3 \rxn 5A$
    for the positive terms and
    $3X_2 \rxn B$ for the negative term.
    The reactions of \Cref{thm:majority} then properly decide whether $A>B$,
    which,
    since the above reactions converge to $A=2x_1+\frac{5}{4} x_3$ and $B=4+x_2/3$,
    is true if and only if the threshold predicate 
    $2x_1 - x_2/3 + \frac{5}{4} x_3 > 4$
    holds.
\end{proofsketch}
}
\begin{proof}
Intuitively, this construction translates the threshold problem to a majority problem. We translate the positive and negative contributions in the weighted sum to concentrations of species $A$ and $B$ respectively. We then use the reactions from the majority problem to detect if $A > B$ (with the threshold $h$ added appropriately to either $A$ or $B$).

Let $\phi:\mathbb{R}_{\ge 0}^k \to \{0,1\}$ be a threshold predicate.
Without loss of generality that we assume all weights 
are integers by clearing denominators,
so $\phi(\vx) = 1 \iff \sum_{i=1}^k w_i \vx(i) > h$
with each $w_i \in \Z$ and $h \in \R$.
We construct the CRD $\calD = (\Lambda,R,\Sigma,\yesVotes\noVotes,\vi)$ as follows. 
$\Sigma = \{ X_1,\dots,X_k \}$,
and $\Lambda$ is implicitly all species described below.
Let $\Upsilon_\mathrm{yes} = \{Y\}$ and $\Upsilon_\mathrm{no} = \{ N \}$ be the yes and no voter species respectively. Let the $\vi = \{ 1 Y, 1 C, |h| AB\}$,
where $AB = A$ if $h > 0$ and $AB = B$ otherwise. We then add the following reactions. For any $w_p > 0$ add reaction $X_p \rxn w_p A$ and for any $w_n < 0$ add reaction $X_n \rxn |w_n| B$. We also add the reactions from the CRD computing majority in \cref{thm:majority}.

Without loss of generality, we may insist that the threshold value $h$ is zero as the reactions $L \to hB$ and $L \to |h|A$ influence the concentration of $A$ and $B$ respectively in the ``same way'' as the reactions $X_i \to w_i(A/B)$ in the sense that the concentration of each $X_i$ species and $L$ have closed form solutions of decaying exponential functions. Let $\mathcal{P}$ be the set of all $p$ for which $w_p > 0$ and let $\mathcal{N}$ be the set of all $n$ for which $w_n < 0$. Let $a = \sum_{p \in \mathcal{P}} w_p \cdot x_p$ and $b = \sum_{n \in \mathcal{N}} |w_n| \cdot x_n$. We consider three cases.

\begin{description}
\item[$a>b$:] 
    It suffices to observe that
    $
    \lim_{t \to \infty} A(t) = a-b
    $
    and 
    $
    \lim_{t \to \infty} B(t) = 0
    $
    since the restricted CRN that only contains species that change the concentration of $A$ and $B$ is feed-forward.
    A similar argument to the one used to prove the $a > b$ case of \cref{thm:majority} then shows the result. 

\item[$a<b$:]
    Symmetric to the previous case.

\item[$a=b$:]
    Our goal is to show that the concentration of species $N$ converges to 1. 
    Equivalently, we will show that the concentration of the yes voter $Y$ converges to zero as the invariant $Y(t) + N(t) = 1$ is held for all $t$. 
    The following argument will be very similar to the one we used to prove the correctness of our majority CRD. 
    The CRD induces the following ODE for the concentration $Y(t)$:
    \begin{equation}
        Y'(t) = k_1A(t)N(t)-k_2B(t)Y(t)-k_4C(t)Y(t)
    \end{equation}
    Using the fact that $Y(t) + N(t) = 1$ we can rewrite this ODE as
    \begin{equation}
        Y'(t)+Y(t)[k_1A(t)+k_2B(t)+k_4C(t)] =k_1A(t).
    \end{equation}
    From the initial value problem $C'(t) = -C(t)^3$ with $C(0) = 1$, we obtain that a closed form for the concentration of $C$ is $C(t) = (2k_4t+1)^{-1/2}$. Hence $C(t)$ is strictly positive, so we can apply \cref{lem:linear-odes} to find that 
    \begin{align}
        Y(t) &\le \frac{k_1A(t)}{k_1A(t)+k_2B(t)+k_4C(t)}+K\exp\pars{-\int_0^t k_1A(s)+k_2B(s)+k_4C(s)ds}
    \end{align}
    With the observation that $B(t) \ge 0$ and $A(t) \ge 0$ for all $t \in \R_{\ge 0}$, and substituting in the closed form for $C(t)$, we obtain the bound
    \begin{align}
        Y(t) &\le \frac{k_1A(t)}{k_4C(t)}+K\exp\left(-(2k_4^2t+k_4))\right)
    \end{align}
    The exponential term clearly converges to zero, so to prove the result it suffices to show the quotient term converges to zero. We observe that the restricted CRN with only reactions that change the concentration of $A$ and $B$ is the feed-forward CRN:
    \begin{eqnarray*}
        \forall p \in \mathcal{P}: X_p &\rxn& w_p A\\
        \forall n \in \mathcal{N}: X_n &\rxn& |w_n| B\\
        A + B &\rxn& \emptyset
    \end{eqnarray*}
    This CRN stably computes the function $f(a,b) = \min(a-b,0)$, which by \cite{chen2023rate} Corollary 4.11 implies that it robustly computes $\min(a-b,0)$. This shows that $\lim_{t\to \infty}A(t) = a-b = 0$. Taking the limit of the quotient, we obtain $0/0$, a so-called indeterminate form. Hence, we may apply L'Hôpital's rule to obtain the following estimate.
    \begin{align*}
        \lim_{t \to \infty} \frac{k_1}{k_4}\frac{A(t)}{C(t)} 
    &= 
        \frac{k_1}{k_4}\lim_{t \to \infty} \frac{A(t)'}{C(t)'}
    \\ &=
        \frac{k_1}{k_4}\lim_{t \to \infty} \frac{\sum_{p \in \mathcal{P}}k_p  X_p-k_3A(t)B(t)}{(2k_4t+1)^{-3/2}}
    \\ &\le
        \frac{k_1}{k_4}\lim_{t \to \infty} \frac{\sum_{p \in \mathcal{P}}k_p  X_p}{(2k_4t+1)^{-3/2}}
    \end{align*}
    The equation for $A'(t)$ is from its mass action differential equation. Each input species $X_i$ can be found to have a closed form of $w_ix_ik_ie^{-k_p t}$, which shows that the quotient tends to zero as $t$ goes to infinity. Thus, $\lim_{t\to\infty} Y(t) = 0$ as desired.
\qedhere
\end{description}
\end{proof}

We will now show that Boolean combinations of threshold predicates are also robustly decidable. To show this result, we first prove a lemma that lets us assume CRDs that robustly compute predicates are of a convenient form. That is, we will now show that it is without loss of generality to assume that a CRD has exactly one yes and no voter, and furthermore both of their concentrations converge exactly as $t \to \infty$.

\begin{lemmarep}
    \label{lem:twoVoters}
    Let $\calD = (\Lambda,R,\Sigma,\yesVotes,\noVotes,\vi)$ robustly decide the predicate $\phi:\domR^k \to \{0,1\}$. 
    Then there is a CRD $\calD'=(\Lambda',R',\Sigma',\{Y\},\{N\},\vi')$ that robustly decides $\phi$ with exactly one yes voter species $Y$ and one no voter species $N$. 
    Furthermore, the concentration of these voters satisfy $Y(t) + N(t) = 1$ for all $t \in \R_{\ge 0}$.
\end{lemmarep}

\begin{proofsketch}
    A full proof is given in \Cref{apx:robust_predicate_computation}.
    For each original yes voter $V_\mathrm{y}$, we add a reaction $V_\mathrm{y} + N \rxn V_\mathrm{y} + Y$,
    and similarly 
    $V_\mathrm{n} + Y \rxn V_\mathrm{n} + N$,
    to influence the new voters in the correct direction.
\end{proofsketch}

\begin{proof}
We construct the CRD $\calD'$ as follows:
\begin{enumerate}
    \item Let $\Sigma' = \Sigma $, $\Lambda' = \Lambda \cup\{Y,N\}$.
    
    \item Keep all reactions from the original CRD, and for each yes voter $\Vy \in \yesVotes$, add the reaction 
    \begin{equation}
        \label{eq:one_voter_yes}
        \Vy + N \rxn^{k_{\Vy}} \Vy + Y.
    \end{equation}
    
    Additionally, for each no voter $\Vn \in \noVotes$ add the reaction
    \begin{equation}
        \label{eq:one_voter_no}
        \Vn + Y \rxn^{k_{\Vn}} \Vn + N.
    \end{equation}
    
    
    \item Start with initial context $\vi' = \vi+ \{1Y\}$.
\end{enumerate}

To show that $\calD'$ robustly decides $\phi$, we will show that for all $\vx \in \domR^k$ such that $\phi(\vx) = 1$, the concentration of species $Y$ approaches 1 as $t$ approaches infinity when the initial state is $\vx + \vi + \{1Y\}$. The case when $\phi(\vx) = 0$ is a similar argument. 

We make the worst-case assumption that the adversary chooses $\ky = \min\limits_{\Vy\in\yesVotes} k_{\Vy}$ to be the same (and smallest) rate constant for all reactions of the form \eqref{eq:one_voter_yes}, and similarly chooses the rate constant $\kn = \max\limits_{\Vn\in\noVotes} k_{\Vn}$ for all reactions of the form \eqref{eq:one_voter_no}.
This is a safe worst-case assumption, since choosing larger rate constants for reaction \eqref{eq:one_voter_yes} or smaller rate constants for \eqref{eq:one_voter_no} would only make the network converge faster to the correct values.
Furthermore, scaling all rate constants in the system by the same factor modifies only the timescale of the mass-action trajectory of the ODEs, while preserving its shape and eventual convergence. To simplify the system, we scale all rate constants by $1 / \ky$, ensuring that $\ky = 1$.

The CRD $\calD'$ induces the following ODE for the concentration of the species $N$
\[
N'(t)+N(t)(\kn\Vn(t)+\ky\Vy(t)) = \kn\Vn(t)
\]
We may then apply \cref{lem:linear-odes-zero} to obtain the bound 
\[
N(t) \le \frac{5\kn\Vn(t)}{\kn\Vn(t) + \ky\Vy(t)+(t^2+1)^{-1}} + K\exp\pars{-\int_0^t\kn\Vn(s)+\ky\Vy(s)ds}.
\]
By the correctness of $\calC$, the concentration $\Vy(t)$ remains strictly positive as $t \to \infty$. Hence, the integral $\int_0^t \kn\Vn(s)+\ky\Vy(s)~ds$ diverges to infinity as $t \to \infty$. This demonstrates that the exponential term tends to zero as $t$ grows large. To see the quotient term goes to zero, the correctness of $\calC$ also dictates that $\Vn(t) \to 0$ as $t \to \infty$. This shows that $N(t) \to 0$.
\end{proof}

This result shows that \cref{def:robustly-decide} is equivalent to a model in which we require exactly one yes voter and no voter.
Furthermore, this lemma allows us to insist that correct voting species do not just remain above zero, but in fact converge to a particular value. 

\begin{lemma}
\label{lem:stably-decidable-closed-Boolean}
    Let $\phi_1:\mathbb{R}_{\ge 0}^k \to \{0,1\}$ and $\phi_2:\mathbb{R}_{\ge 0}^k \to \{0,1\}$ be robustly decidable predicates.
    Then the following are also robustly decidable: $\overline{\phi_1}$, $\phi_1 \land \phi_2$, and $\phi_1 \lor \phi_2$.
\end{lemma}
\begin{proof}
We will use the same construction as \cref{lem:stable-decidability-closure-properties}, restated here for convenience.

\noindent
\textbf{Complement:}
If $\phi_1$ is robustly decidable by $\calD_1 = (\Lambda_1,R_1,\Sigma,\yesVotes,\noVotes,\vi_1)$, we can robustly decide $\overline{\phi}_1$ with the CRD $\overline{\calD}_1 = (\Lambda_1,R_1,\Sigma,\noVotes,\yesVotes,\vi_1)$. That is, it is enough to switch yes voters and no voters to robustly compute $\overline{\phi_1}$. 

\noindent
\textbf{Union and Intersection:}
Suppose $\phi_1$ and $\phi_2$ are respectively robustly decided by the CRDs $\calD_1 = (\Lambda_1,R_1,\Sigma,\{Y^1\},\{N^1\},\vi_1)$ and $\calD_2 = (\Lambda_2,R_2,\Sigma,\{Y^2\},\{N^2\},\vi_2)$. By \cref{lem:twoVoters}, we may assume that each $\calD_i$ has exactly two voters ($\{Y^i,N^i\}$), and maintains the invariant that it the sum of the concentrations of these voters is 1 ($Y^i(t) + N^i(t) = 1$ for all $t \in \domR$).

We construct a CRD $\calD = (\Lambda,R,\Sigma,\yesVotes,\noVotes,\vi)$ with 
$\Lambda = \Lambda_1 \cup \Lambda_2 \cup \{\vyy,\vyn,\vny,\vnn\}$. The added species \vyy, \vyn, \vny, and \vnn\ represent all four possible outputs of two threshold predicates. 
To decide $\phi_1 \lor \phi_2$,
let $\Upsilon_\mathrm{yes} = \{\vyy, \vny,\vyn\}$ and $\Upsilon_\mathrm{no} = \{\vnn\}$ and to decide $\phi_1 \land \phi_2$, 
let $\Upsilon_\mathrm{yes} = \{\vyy\}$ and $\Upsilon_\mathrm{no} = \{\vny, \vyn, \vnn\}$. 
Let $\vi = \vi_1 + \vi_2 + \{\vyy\}$.

We add the following reactions: For each input species $X_i$, add the reaction $X_i \rxn X_i^1 + X_i^2$ where $X_i^1$ and $X_i^2$ are the input species of $\calD_1$ and $\calD_2$ respectively.

To ``record'' the votes of $\calD_1$ and $\calD_2$ in the voter species, we include the following reactions. 
Without loss of generality assume that the adversary chooses the same \emph{minimum} rate constant $s = 1$ for reactions in which a former yes voter is a catalyst, and the same maximal rate constant $f\geq s$ for reactions where a former no voter is a catalyst. 
That is, include reactions
\begin{align}
Y^1 + \vnn \rxn^f Y^1 + \vyn
\label{rxn-comb-1}
\\
\Sy^{1} + \vny \rxn^f Y^1 + \vyy
\label{rxn-comb-2}
\\
N^1 + \vyn \rxn^1 N^1 + \vnn
\label{rxn-comb-3}
\\
N^1 + \vyy \rxn^1 N^1 + \vny
\label{rxn-comb-4}
\end{align}
where $Y^1$ is the yes voter from $\calD_1$ and $N^1$ is the no voter. Similarly, include reactions:
\begin{align}
Y^2 + \vnn \rxn^f Y^2 + \vny
\label{rxn-comb-5}
\\
Y^2 + \vyn \rxn^f Y^2 + \vyy
\label{rxn-comb-6}
\\
N^2 + \vny \rxn^1 N^2 + \vnn
\label{rxn-comb-7}
\\
N^2 + \vyy \rxn^1 N^2 + \vyn
\label{rxn-comb-8}
\end{align}
where $Y^2$ and $N^2$ are the voting species from $\calD_2$.

For analysis, it suffices to show that when $\phi_1(\vx)=\phi_2(\vx)=0$ on input $\vx$, the concentration of no voter $\vnn(t)$ approaches 1 and the other concentration of the three voters approach 0. Since $\vnn + \vny + \vyn + \vyy = 1$, it this is equivalent to simply showing that $\vnn$ converges to 1.
To do this, we will demonstrate that for every $0< \varepsilon <1$, there exists $t_0 > 0$ such that for every $t>t_0$, $\vnn(t) \geq 1-\varepsilon$
(thus, for each $V \in \{\vnn,\vyn,\vny\}$,
$V(t) \leq \varepsilon$, since all four concentrations sum to 1).
Let $\delta = \varepsilon/(2f)$.
By the correctness of $\calD_1$ and $\calD_2$, there is a $t_0'$ such that, for all $t > t_0'$, 
$N^i(t) \geq1-\delta$, and $Y^i(t) \leq \delta$ for all $t > t_0'$. We now argue that at some later $t_0 > t_0'$, 
$\vnn(t) \geq 1-\varepsilon$ for all $t > t_0$.



Since the CRD only converges faster if the $N^i$'s are larger and the $Y^i$'s smaller, 
we make the worst case assumption that for all $t > t_0$,
$N^i(t) = 1 - \delta$ and $Y^i(t) = \delta.$
Under this assumption, where one reactant of each reaction does not change concentration over time,
we can analyze it as a unimolecular reaction with a different rate constant.
For example, since the rate of $Y^2 + \vyn \rxn^f Y^2 + \vyy$ is $\vyn(t) \cdot \delta \cdot f$,
we treat it as the unimolecular $\vyn \rxn^{\delta\cdot f} \vyy.$

After this change, the species $\vyn$ and $\vny$ become functionally equivalent in the sense that, the reactions $\vyn \rxn^{\delta \cdot f} \vyy$ and $\vny \rxn^{\delta \cdot f} \vyy$ can be replaced by the single reaction $\vm \rxn^{2 \cdot \delta \cdot f} \vyy,$
where $\vm$ is a fictional species whose concentration is the sum of the concentrations of $\vyn$ and $\vny$.
We can combine the following pairs of reactions in this way:
    \eqref{rxn-comb-1} and \eqref{rxn-comb-5},
    \eqref{rxn-comb-2} and \eqref{rxn-comb-6},
    \eqref{rxn-comb-3} and \eqref{rxn-comb-7},
    \eqref{rxn-comb-4} and \eqref{rxn-comb-8}.
After these changes to a ``worst case'' CRD whose concentration of no voters (resp. yes voters) lower bound (resp. upper bound) the concentration of species in the real CRD, the system becomes (after simplifying by removing the factor 2 from all rate constants, again without altering the trajectory other than its timescale)
$
\vnn 
\rxn^{\delta \cdot f} \vm,
\quad
\vm 
\rxn^{1-\delta } \vnn,
\quad
\vm 
\rxn^{\delta  \cdot f} \vyy,
\quad
\vyy
\rxn^{1-\delta } \vm,
$
written more succinctly as
$\vyy 
\revrxn^{1-\delta }_{\delta  \cdot f}
\vm
\revrxn^{1-\delta }_{\delta  \cdot f}
\vnn.$
Plugging in the definition $\delta = \varepsilon/(2f)$, these reactions become:
$
\vyy 
\revrxn^{1-\varepsilon/(2f)}_{\varepsilon/2}
\vm
\revrxn^{1-\varepsilon/(2f)}_{\varepsilon/2}
\vnn.
$
Since in the first normalization of rate constants so that the smallest was 1, and $f$ was the maximum, we have $f \geq 1$, 
so for $\varepsilon < 1$, we have $1-\varepsilon/(2f) \geq 1/2$.
Again making the worst-case assumption that the rate constants for the forward reactions above are at this minimal bound, we have that the following will converge no faster than the real network:
$
\vyy 
\revrxn^{1/2}_{\varepsilon/2}
\vm
\revrxn^{1/2}_{\varepsilon/2}
\vnn.
$
We can simplify the above by scaling the rate constants by 2 and get $\vyy 
\revrxn^{1}_{\varepsilon}
\vm
\revrxn^{1}_{\varepsilon}
\vnn$,
which has the ODEs
\begin{align*}
    \vyy' &= \varepsilon \vm - \vyy\\
    \vm' &= \vyy + \varepsilon \vnn -(\varepsilon+1) \vm\\
    \vnn'&= \vm - \varepsilon \vnn
\end{align*}
This is a system of first-order,
\todo{KC: We maybe want to reword this from ``this system converges no faster than the original'' to ``the concatenation of \vnn in this system lower bounds the concentration of \vnn in the original system'' since technically \vnn doesn't converge to 1 in this system}
linear differential equations and can be found to have solutions $\vyy(t),\vm(t)$ and $\vnn(t)$ satisfying
\[
\lim_{t \to \infty}\vyy(t) = \frac{\varepsilon^2}{\varepsilon^2+\varepsilon+1},
\quad \lim_{t \to \infty}\vm(t) = \frac{\varepsilon}{\varepsilon^2+\varepsilon+1},
\quad \lim_{t \to \infty}\vnn(t) = \frac{1}{\varepsilon^2+\varepsilon+1}
\]

One can confirm that $1 - \frac{1}{\varepsilon^2+\varepsilon+1} < \epsilon$, i.e., in this simplified system, \vnn\ converges to a value strictly within $\epsilon$ of its claimed convergent value of 1 in the original system, thus gets within, and stays within, $\epsilon$ of 1 after some finite time $t_0$, as we intended to show.
\end{proof}

The following is the first main result of this paper.

\begin{theorem}
\label{thm:multi-threshold-robustly-decidable}
Every multi-threshold predicate is robustly decidable by
a continuous CRD.
\end{theorem}

\begin{proof}
Immediate from induction on the number of threshold predicates defining the multi-threshold predicate; 
\Cref{thm:threshold-robustly-decidable}
is the base case and
\Cref{lem:stably-decidable-closed-Boolean}
is the inductive case.
\end{proof}








\subsection{Piecewise affine functions are robustly computable}
\label{sec:robust_function_computation}



The next definition captures the idea that a species $S$ converges because reactions stop changing $S$, as opposed to reactions producing $S$ at the same rate other reactions consume $S$.

\begin{definition}
\label{def:approaches-static-equilibrium}
Consider CRN $\calC = (\Lambda, R)$ with initial state $\vx$.
Recall $\rho_t(\alpha)$ is the rate of $\alpha \in R$ at time $t$.
We say that from $\vx$, a species $S \in \Lambda$ 
\emph{approaches static steady state}, if for every reaction $\alpha = (\vr,\vp,k) \in R$ such that $\vr(S) \ne \vp(S)$, 
$\lim_{t\to\infty} \rho_t(\alpha) = 0$.
\end{definition}

Note that if $S$ approaches static steady state then $\lim_{t\to\infty} S(t)$ exists and is finite.
In other words, although the CRN may not converge to a single state
(perhaps some species oscillate or diverge to $\infty$),
not only does the value of $S(t)$ converge, 
a stronger condition holds:
the CRN converges to rate 0 of every reaction net producing or net consuming $S$.
This contrasts \emph{dynamic steady state},
e.g., for $S \revrxn^1_1 A$, where $S$ approaches a limit, but the reactions producing/consuming $S$ have positive rate even at steady state (where $S=A$).

\begin{lemma}
\label{lem:rational-floor-affine-stably-computable}
Every rational floor-affine function $f(\vx) = \max(0, h + \sum_{i=1}^k w_i \cdot \vx(i))$ is stably and robustly computable by a continuous feedforward CRC with output species $Y$,
such that $Y$ approaches static steady state 
and $Y$'s concentration never exceeds $h+\sum_{i, w_i > 0} w_i \cdot \vx(i)$.
\end{lemma}

\begin{proof}
Since all feedforward CRCs that stably compute a function also robustly compute it 
(\Cref{lem:feedforward-stably-compute-implies-robustly-compute}),
it suffices to define a feedforward CRC stably computing $f.$
Define the CRC $\calC = (\Lambda,R,\Sigma,Y,\vi)$,
where $\Sigma = \{X_1,\dots,X_k\}$,
$\Lambda = \Sigma \cup \{Y,Y^-\}$,
and $\vi$ and $R$ are defined below.
Let $f(\vx) = \max(0,g(\vx))$
be rational floor-affine,
with $g: \Rp^k \to \R$ a rational affine function $g(\vx) = h + \sum_{i=1}^k \frac{n_i}{d_i} \cdot \vx(i)$, 
for $h \in \R$, 
$d_1,\dots,d_k \in \Z$, and
$n_1,\dots,n_k \in \Z_{>0}$;
note we have written each rational $w_i$ as $\frac{n_i}{d_i}$.
Start with initial context $\vi = \{h Y\}$ if $h > 0$ and $\vi = \{|h| Y^-\}$ if $h \le 0$.
For every $i$ such that $d_i > 0$,
add reaction $n_i X_i \rxn d_i Y$ to $R$,
and for every $i$ such that $d_i < 0$,
add reaction $n_i X_i \rxn d_i Y^-$ to $R$.
Finally we add reaction $Y+Y^- \rxn \emptyset$. Except for the last reaction, all reactions are entirely independent—none share a reactant with another reaction. 
Since the last reaction does not produce any species, $\calC$ is feedforward. 

From any reachable state, execute reactions of the form 
$n_i X_i \rxn d_i Y$ and 
$n_i X_i \rxn d_i Y^-$ until all inputs are gone.
By inspection of the reactions,
this means $\sum_{i, n_i > 0} \frac{n_i}{d_i} \vx(i)$ amount of $Y$ has been produced in total by such reactions,
and $\sum_{i, n_i < 0} \frac{n_i}{d_i} \vx(i)$ amount of $Y^-$ has been produced in total by such reactions.
Then run reaction $Y + Y^- \to \emptyset$ until both reactants are gone.
If $g(\vx) > 0$ then $Y^-$ will be limiting and this will result in 
$
\sum_{i, n_i > 0} \frac{n_i}{d_i} \vx(i) - 
\sum_{i, n_i < 0} \frac{n_i}{d_i} \vx(i)
=
\sum_{i=1}^k \frac{n_i}{d_i} \vx (i)
= g(\vx) = f(\vx),
$
and if $g(\vx) \leq 0$ then $Y$ will go to 0, so $\calC$ stably computes $f(\vx) = \max(0, g(\vx)).$

Since $\calC$ is feedforward, its steady state is static~\cite[Lemma 4.8]{chen2023rate}.
Finally, the fact that $Y$ cannot exceed 
$h + \sum_{i, w_i > 0} w_i \cdot \vx(i)$
follows from the fact that we start with $h Y$ (if $h > 0$)
and the only reactions producing $Y$ are
$n_i X_i \rxn d_i Y$ for $i$ such that $\frac{n_i}{d_i} = w_i > 0$.
\end{proof}

The following is the second main result of this paper.
\begin{theorem}
\label{thm:piecewise-affine-are-robustly-computable}
Every threshold-piecewise rational floor-affine function is robustly computable by a continuous CRC.
\end{theorem}

\begin{proof}
    This construction is similar to Lemma 4.4 in~\cite{CheDotSolNaCo},
    though that proof was different (and simpler),
    applying to stable computation in the discrete CRN model.
    For a threshold-piecewise rational floor-affine function $f:\Rp^k \to \Rp$,
    recall the floor-affine functions $f_1,\dots,f_l: \Rp^k \to \Rp$
    and the multi-threshold predicates $\phi_1,\dots,\phi_\ell: \Rp^k \to \{0,1\}$
    from \Cref{def:piecewise-affine}.
    We equivalently think of each $\phi_j$ as a set $R_j = \phi_j^{-1}(1)$ that makes $\phi_j$ true,
    i.e., 
    $f$ partitions the $\Rp^k$ into $\ell$ disjoint regions $R_1,\dots,R_\ell$.

    \newcommand{\D}{{\mathcal{D}}}
    \newcommand{\C}{{\mathcal{C}}}
    Let $j \in \{1,\dots,\ell\}$.
    By \Cref{lem:stably-decidable-closed-Boolean}
    there is a CRD 
    $\calD_j = (\Lambda_j^\D, R_j^\D, \Sigma_j^\D, \{T_j\}, \{F_j\}, \vi_j^\D)$,
    robustly deciding $\phi_j$,
    where we may assume a single yes voter $T_j$ and single no voter $F_j$ for each by \Cref{lem:twoVoters}.
    By \Cref{lem:rational-floor-affine-stably-computable}
    there is a CRC 
    $\calC_j = (\Lambda_j^\C, R_j^\C, \Sigma_j^\C, Y_j, \vi_j^\C)$ 
    robustly computing $f_j$.
    We construct the CRC $\mathcal{C}=(\Lambda, R, \Sigma, Y^P, \vi)$ to robustly compute $f$,
    where $\Sigma = \{X_1,\ldots,X_k\}$.
    We transform the input species from $\Sigma$ to those of each $\Sigma_i^\D, \Sigma_i^\C$ 
    by adding reactions of the form,
    for each $1 \leq i \leq k$,
    $X_i \rxn X_i^{1,\D} + X_i^{1,\C} + \dots + X_i^{\ell,\D} + X_i^{\ell,\C}$ 
    (similar to the construction in the proof of \cref{lem:stably-decidable-closed-Boolean}).
    
    Ideally we would convert the correct ``local'' output $Y_j$ to the ``global'' output $Y$, but $\calC_j$ may need to consume $Y_j$; 
    in fact it provably must do so if computing a non-monotone function such as $x_1 - x_2.$
    To avoid interfering with $\calC_j$'s computation,
    we instead modify its reactions to have additional products that will be used in the conversion.
    Since we do not change any reactants, and since the new products are new species not appearing in $\calC_j$,
    these new products will preserve the kinetics of the CRC $\calC_j$,
    while enabling the constructed CRC to ``read'' the output of $\calC_j$.\footnote{
        We use a trick known as \emph{dual-rail encoding}~\cite{CheDotSolNaCo,chen2013programmable};
        we use it more like \cite{CheDotSolNaCo} as a handy ``intermediate'' proof technique, not as in \cite{chen2023rate} in which inputs and outputs of CRCs are encoded in dual-rail.
    }

    For each reaction net producing $p$ copies of $Y_j$ in $\calC_j$,
    for example 
    $A+Y_j \rxn B+3Y_j$ that net produces 2 $Y_j$,
    we modify the reaction to also produce $p$ copies of a new species $Y_j^P$:
    $A+Y_j \rxn B+3Y_j + 2Y_j^P$.
    For each reaction net consuming $c$ copies of $Y_j$ in $\calC_j$,
    for example 
    $B+4Y_j \rxn Y_j$ that net consumes 3 $Y_j$,
    we modify the reaction to also produce $c$ copies of a new species $Y_j^C$:
    $B+4Y_j \rxn Y_j + 3Y_j^C$.
    This maintains that $Y_j(t) = Y_j^P(t) - Y_j^C(t)$ for all $t$.
    $Y_j^P$ ``records'' the total amount of $Y_j$ that has ever been produced by any reaction,
    and $Y_j^C$ the total amount of $Y_j$ that has ever been consumed,
    so that their difference is the net amount of $Y_j$ produced.
    If $\vi_j^\C(Y_j) > 0$,
    we also modify $\vi$ so that $\vi(Y_j^P) = \vi_j(Y_j)$,
    i.e., if $\calC_j$ starts with some $Y_j$ already ``produced'', we want to start with the same amount of $Y_j^P$ to maintain 
    $Y_j(t) = Y_j^P(t) - Y_j^C(t)$.
    (captured by the multiset term $\{\vi_j^\C(Y_j)\ Y_j^P\}$ below)
    
    Let 
    $
    \Lambda = \{Y^P, Y^C \} \cup \Sigma \cup \bigcup_{j=1}^\ell \Lambda_j^\D \cup \Lambda_j^\C \cup \{Y_j^P,Y_j^C\}, $
    and 
    $\vi = \sum_{j=1}^\ell \vi_j^\calD + \vi_j^\calC + \{\vi_j^\C(Y_j)\ Y_j^P\}.$
    
    $R$ includes all reactions from each $\calC_j$ and $\calD_j$,
    after modifying some as explained above,
    and the following reactions for each $1 \leq j \leq \ell$:

    \begin{eqnarray}
        T_j + Y_j^P &\rxn^{k_{j_1}}& T_j + \hat{Y}_j^P + Y^P
        \label{rxn:Tj_activate_YjP}
    \\
        T_j + Y_j^C &\rxn^{k_{j_2}}& T_j + \hat{Y}_j^C + Y^C
        \label{rxn:Tj_activate_YjC}
    \\
        F_j + \hat{Y}_j^P &\rxn^{k_{j_3}}& F_j^P + Y_j^P  + Y^C
        \label{rxn:Fj_deactivate_YjP}
    \\
        F_j + \hat{Y}_j^C &\rxn^{k_{j_4}}& F_j^C + Y_j^C  + Y^P
        \label{rxn:Fj_deactivate_YjC}
    \\
        Y^P + Y^C &\rxn^{k}& \emptyset
        \label{rxn:Yp-Yc-to-nothing}
    \end{eqnarray}
    
    We think of the outputs $Y_i^P, Y_i^C$ produced by $\calC_i$ as ``inactive'',
    with ``active'' versions $\hat Y_i^P, \hat Y_i^C$.
    Then $T_j$ activates them and $F_j$ deactivates them,
    while maintaining that as much $Y^P,Y^C$ is produced as the amount of $\hat Y_j^P, \hat Y_j^C$ that has been net activated by the voters $T_j,F_j$.
    Furthermore, to ease analysis of the ODEs,
    we ensure that $Y^P$ and $Y^C$ do not influence the reaction rates.
    For, say, reaction \eqref{rxn:Fj_deactivate_YjP} to reverse the effect of reaction \eqref{rxn:Tj_activate_YjP} straightforwardly,
    it would consume $Y^P$ as a reactant.
    However, with dual rail it is equivalent to produce $Y^C$, simplifying the expression for the rate of \eqref{rxn:Fj_deactivate_YjP}.

    By inspection of reactions
    \eqref{rxn:Tj_activate_YjP},
    \eqref{rxn:Tj_activate_YjC}, 
    \eqref{rxn:Fj_deactivate_YjP}, \eqref{rxn:Fj_deactivate_YjC},
    \eqref{rxn:Yp-Yc-to-nothing},
    as well as the modified reactions of $\calC_j$,
    one can verify the conservation law that for all $t$,
    \begin{equation}
    \label{eq:conservation_law1}
        \sum_{j=1}^\ell \hat Y_j^P(t) 
        -
        \sum_{j=1}^\ell \hat Y_j^C(t) 
        = Y^P(t) - Y^C(t).
    \end{equation}
    since each reaction changes each side of \eqref{eq:conservation_law1} by the same amount.
    Similarly, the following conservation law holds for each $1 \leq j \leq \ell$:
    \begin{equation}
    \label{eq:conservation_law2}
          Y_j^P(t) - Y_j^C(t) 
        + \hat Y_j^P(t) - \hat Y_j^C(t)
        = Y_j(t).
    \end{equation}
    This completes the description of the construction of $\calC$.

    We now prove that $\calC$ robustly computes $f$.
    Let $\vx \in \Rp^k$.
    Assume without loss of generality that $\phi_1(\vx) = 1$
    (thus $\phi_j(\vx)=0$ for all $j>1$),
    i.e.,
    that $f(\vx) = f_1(\vx)$,
    so that our goal is to converge to activating all 
    $\hat Y_1^P, \hat Y_1^C$
    and deactivating all
    $\hat Y_j^P, \hat Y_j^C$ for $j>1.$
    Since each $\calD_j$ robustly decides $\phi_j$,
    $\lim_{t\to\infty} T_1 = 1$,
    $\lim_{t\to\infty} F_1 = 0$,
    and for all $j>1$,
    $\lim_{t\to\infty} T_j = 0$,
    $\lim_{t\to\infty} F_j = 1$.


    By the conservation law shown above \eqref{eq:conservation_law1}, to demonstrate that $\lim_{t \to \infty} Y^P(t) = f_1(\vx)$,
    it suffices to show that $\hat{Y}_1^P(t)-\hat{Y}_1^C(t)$ approaches $f_1(\vx)$ while $\hat{Y}_j^P(t)$ and $\hat{Y}_j^C(t)$ approach 0 for all $j > 1$.
    We first will show that for all $j>1$, $\hat{Y}_j^P(t)$ and $\hat{Y}_j^C(t)$ approach 0. 
    For each $j$, the ODE for $\hat{Y}_j^P$ is given by
    $
    \frac{d}{dt}\hat{Y}_1^P = k_{j_1}T_jY_j^P-k_{j_3}F_j\hat{Y}_j^P.
    $
    We apply \cref{lem:linear-odes-zero} to obtain that $\hat{Y}_1^P$ is bounded as
    \[
    \hat{Y}_1^P \le \frac{5k_{j_1}T_j(t)Y_j^P(t)}{k_{j_3}F_j(t)+(t^2+1)^{-1}}+Kk_{j_3}\exp\pars{-\int_0^tF_j(s)ds}.
    \]
    By the correctness of $\calD_j$, for all $j > 1$ the concentrations $T_j(t)$ and $F_j(t)$ approach $0$ and $1$ respectively as $t$ approaches infinity. 
    Further, the function $Y_j^P(t)$ is bounded above by the construction in \cref{lem:rational-floor-affine-stably-computable}. 
    Hence, the quotient term converges to zero as $t$ approaches infinity. We also observe that $F_j(t)$ approaching 1 implies that the integral $\int_0^tF_j(s)ds$ diverges as $t \to \infty$, 
    which shows that the exponential term also converges to 0. 
    This demonstrates that $\lim_{t \to \infty} \hat{Y}_j^P(t) = 0$, and a similar argument with the ODE for $\hat{Y}^C_j(t)$ shows that $\lim_{t \to \infty}\hat{Y}^C_j(t) = 0$ as well. 
    To show that $\hat{Y}_1^P(t)-\hat{Y}_1^C(t)$ approaches $f_1(\vx)$, 
    we appeal to the second conservation law derived \eqref{eq:conservation_law2}. Since $\calC_1$ robustly computes $f_1$, the limit $\lim_{t \to \infty}Y_1(t) = f_1(\vx)$. 
    Hence, to show the result we can show simply show that $Y_j^P(t)$ and $Y_j^C(t)$ approach zero. 
    Indeed, the ODE for $Y_j^P(t)$ is given by
    $
    \frac{d}{dt}Y_1^P(t) = k_{3}F_1(t)\hat{Y}_1^P(t)+P_1(\Lambda_1,t)-k_1T_1(t)Y_1^P(t)
    $
    where $P_1(\Lambda_1,t)$ is a function determined by the production of $Y_1^P$ in the CRC $\calC_1$. We apply \cref{lem:linear-odes-zero} to the ODE
    with $p(t) = T_1(t)$ and $g(t) = F_1(t)\hat{Y}_1^P(t) + P(\Lambda_1,t)$ to obtain the bound
    $
    Y_1^P(t) \le \frac{k_{3}F_1(t)\hat{Y}_1^P(t)+P_1(\Lambda_1,t)}{k_1T_1(t)+(t^2+1)^{-1}}+K\exp\pars{-\int_0^tk_1T_1(s)ds}.
    $
    The construction given in \cref{lem:rational-floor-affine-stably-computable} has $Y_1$ approach a static steady state, so the rate of production $P_1(\Lambda_1,t)$ tends to zero as $t \to \infty$. Furthermore, and duel to the previous case, the correctness of $\calD_1$ shows that $F_1(t) \to 0$ and $T_1(t) \to 1$, so $Y_1^P(t) \to 0$. Similar arguments also demonstrate that $Y_1^C(t) \to 0$ as well. To complete the result, we note that by inspection of the reactions, at any time $t$ we have $Y^P(t) \geq Y^C(t)$,
    so reaction \eqref{rxn:Yp-Yc-to-nothing} causes $Y^C(t) \to 0$. Putting it all together, we have shown the desired limit:
    \begin{align*}
        \lim_{t \to \infty}Y^P(t) &= \lim_{t \to \infty}   \pars{    \sum_{j=1}^\ell \hat Y_j^P(t) 
        -
        \sum_{j=1}^\ell \hat Y_j^C(t) + Y^C(t)}\\
        &= \lim_{t \to \infty} \pars{\hat{Y}_1^P(t)-\hat{Y}_1^C(t)} + \lim_{t\to\infty}\pars{\sum_{j=2}^\ell \hat Y_j^P(t) 
        -
        \sum_{j=2}^\ell \hat{Y}_j^C(t)} + \lim_{t\to \infty}Y^C(t)\\
        &= f_1(\vx)+0+0 = f(\vx).
        \qedhere
    \end{align*}
\end{proof}

Note that the discontinuous nature of threshold predicates enables the robust computation of functions that are discontinuous even within the strictly positive orthant;
e.g., \Cref{fig:function}.

%% file: examples.tex
\section{Examples}
\label{sec:examples}


This section shows plots of examples of CRNs that robustly compute some predicates and functions. 
All plots were generated using the Python package gpac~\cite{gpac}.

\subsection{Majority}
Recall the reactions of majority CRD.
See \Cref{fig:majority_lessA,fig:majority_equalAB} for plots of this CRD for varying rate constants and initial values of $A,B$.
\begin{align}
    A + N &\rxn^{k_1} A + Y
    \tag{\ref{rxn:maj-an}}
    \\
    B + Y &\rxn^{k_2} B + N
    \tag{\ref{rxn:maj-by}}
    \\
    A + B &\rxn^{k_3} \emptyset
    \tag{\ref{rxn:maj-ab}}
    \\
    C + Y &\rxn^{k_4} C + N
    \tag{\ref{rxn:maj-cy}}
    \\
    3C &\rxn^{k_5} \emptyset
    \tag{\ref{rxn:maj-3c}}
\end{align}

\begin{figure}[htbp]
\centering
\includegraphics[width=0.5\textwidth]{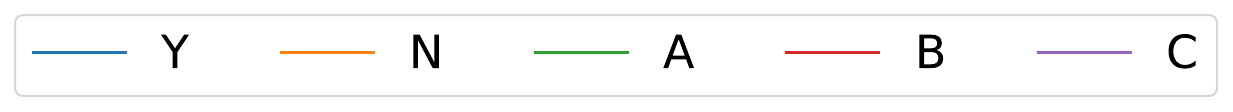}
\\
\begin{subfigure}[b]{0.49\textwidth}
    \centering
    \includegraphics[width=\textwidth]{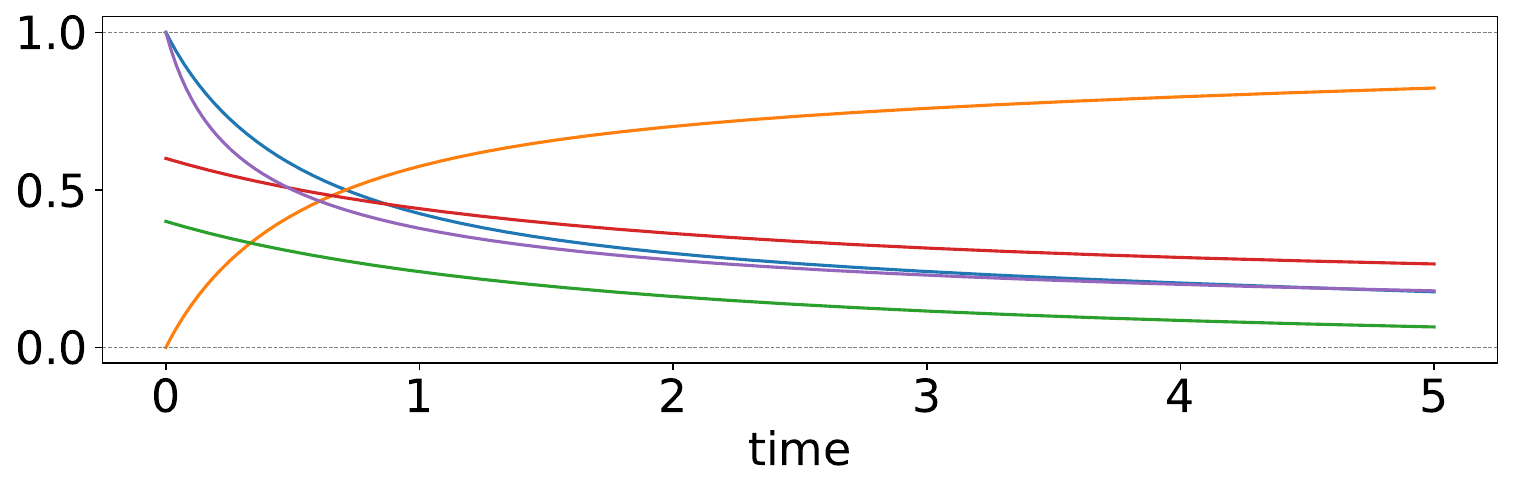}
    \caption{``Short'' time, equal rate constants.}
    \label{fig:majority_lessA_short_equal-rates}
\end{subfigure}
\hfill
\begin{subfigure}[b]{0.49\textwidth}
    \centering
    \includegraphics[width=\textwidth]{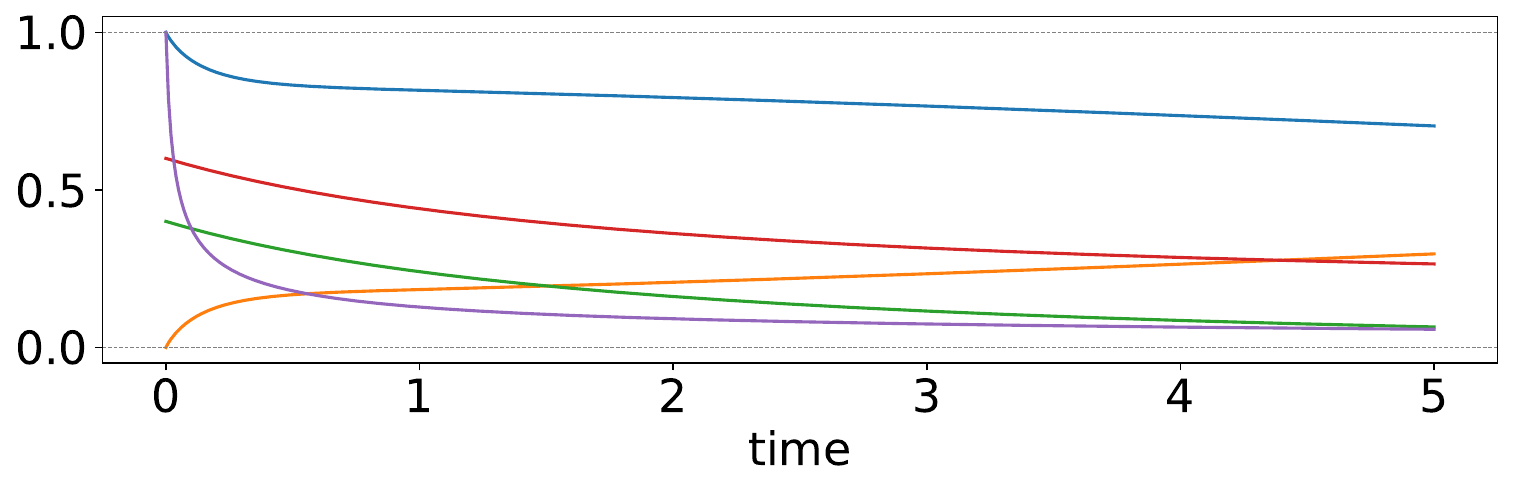}
    \caption{``Short'' time, unequal rate constants.}
    \label{fig:majority_lessA_short_unequal-rates}
\end{subfigure}

\vspace{0.5cm}

\begin{subfigure}[b]{0.49\textwidth}
    \centering
    \includegraphics[width=\textwidth]{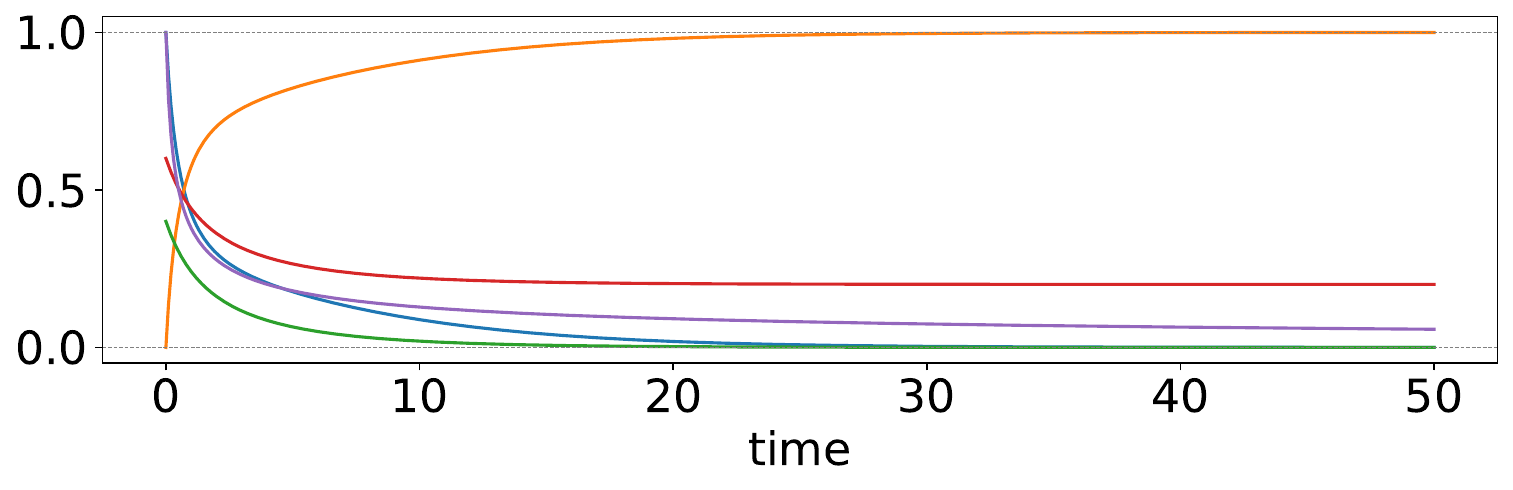}
    \caption{``Long'' time, equal rate constants.}
    \label{fig:majority_lessA_long_equal-rates}
\end{subfigure}
\hfill
\begin{subfigure}[b]{0.49\textwidth}
    \centering
    \includegraphics[width=\textwidth]{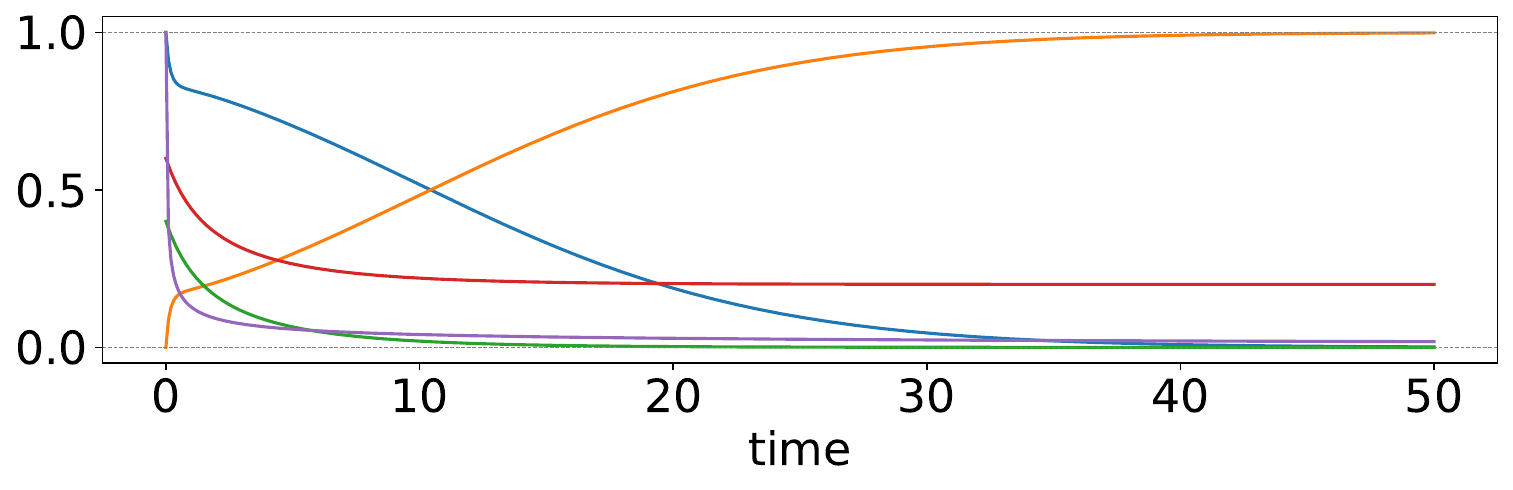}
    \caption{``Long'' time, unequal rate constants.}
    \label{fig:majority_lessA_long_unequal-rates}
\end{subfigure}
\caption{
    The majority CRN of \Cref{thm:majority}, starting with $A=0.4, B=0.6, C=1, Y=1$, shown with different rate constants, simulated for different times to see convergence. 
    Since $A < B$, the correct answer is no, so we should see $N$ converge to 1 and $Y$ to 0. 
    On the left all rate constants are equal to 1. 
    On the right, ``adversarial'' rate constants of reactions 
    \eqref{rxn:maj-an} and \eqref{rxn:maj-3c} are set to 10;
    since the correct answer is no, these reactions ``fight'' correct convergence.
    We set their rate constants larger to demonstrate that the CRD still converges to the correct answer, 
    though not by time 10 
    (\Cref{fig:majority_lessA_short_equal-rates,fig:majority_lessA_short_unequal-rates}).
    By time 100 both have converged
    (\Cref{fig:majority_lessA_long_equal-rates,fig:majority_lessA_long_unequal-rates}).
}
\label{fig:majority_lessA}
\end{figure}

\begin{figure}[htbp]
\centering
\includegraphics[width=0.4\textwidth]{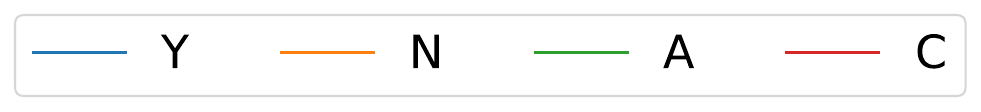}
\\
\begin{subfigure}[b]{0.49\textwidth}
    \centering
    \includegraphics[width=\textwidth]{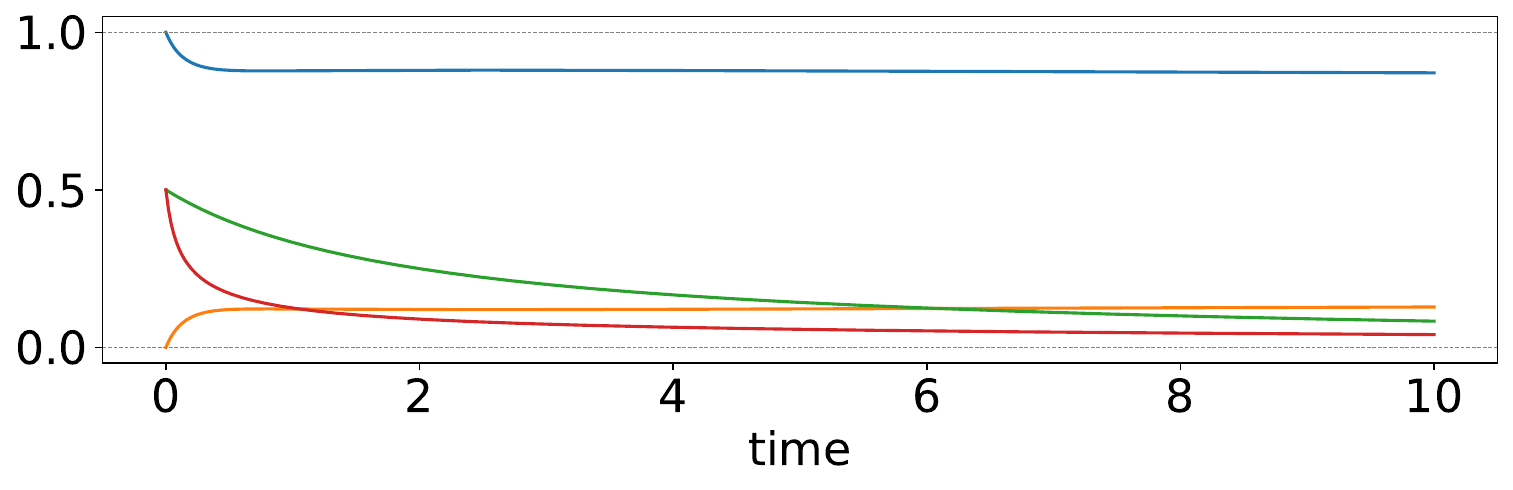}
    \caption{``Short'' time, linear $y$-scale.}
    \label{fig:majority_equalAB_short_unequal-rates}
\end{subfigure}
\hfill
\begin{subfigure}[b]{0.49\textwidth}
    \centering
    \includegraphics[width=\textwidth]{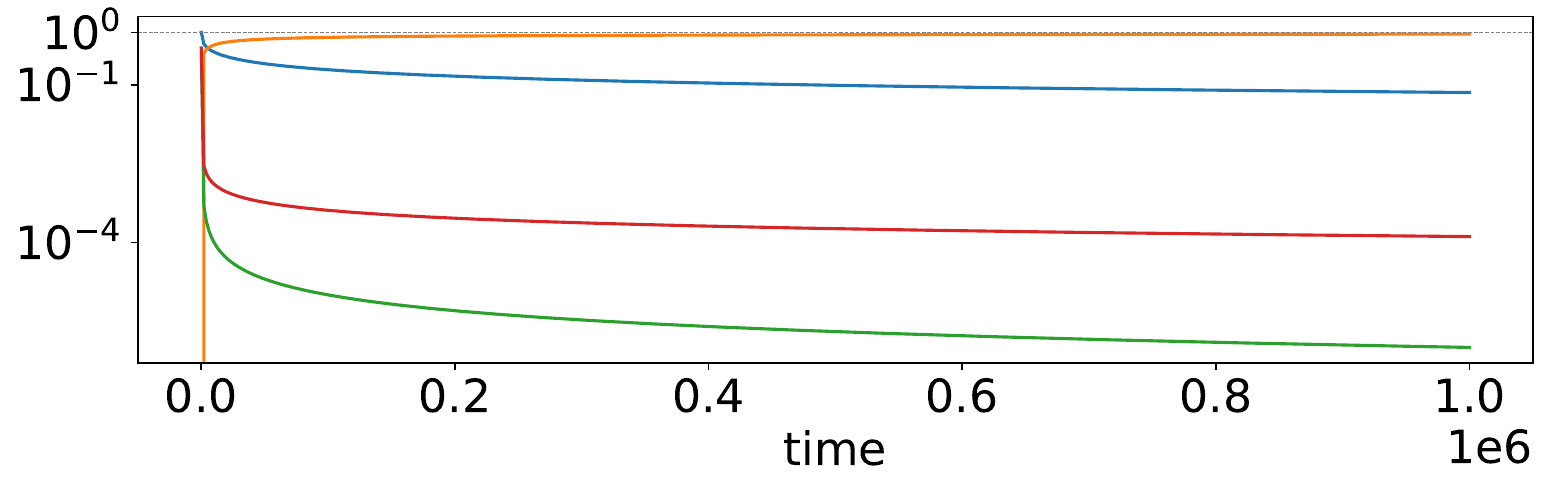}
    \caption{``Medium'' time, log $y$-scale.}
    \label{fig:majority_equalAB_medium_unequal-rates}
\end{subfigure}

\vspace{0.5cm}

\begin{subfigure}[b]{0.49\textwidth}
    \centering
    \includegraphics[width=\textwidth]{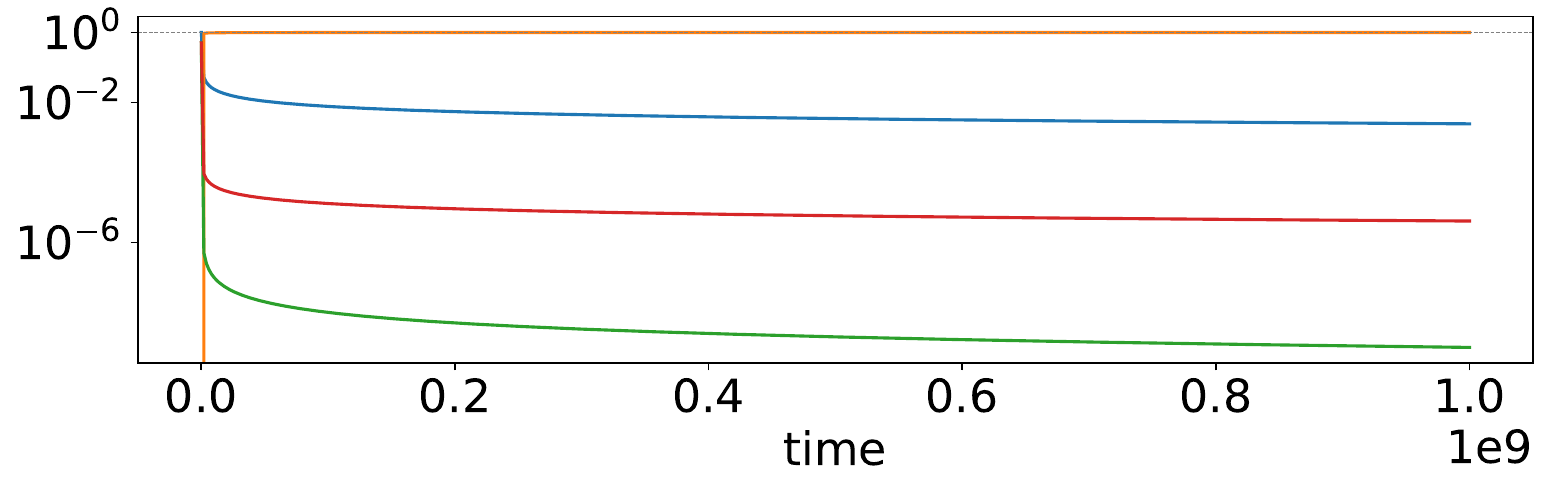}
    \caption{``Long'' time, log $y$-scale.}
    \label{fig:majority_equalAB_long_unequal-rates}
\end{subfigure}
\hfill
\begin{subfigure}[b]{0.49\textwidth}
    \centering
    \includegraphics[width=\textwidth]{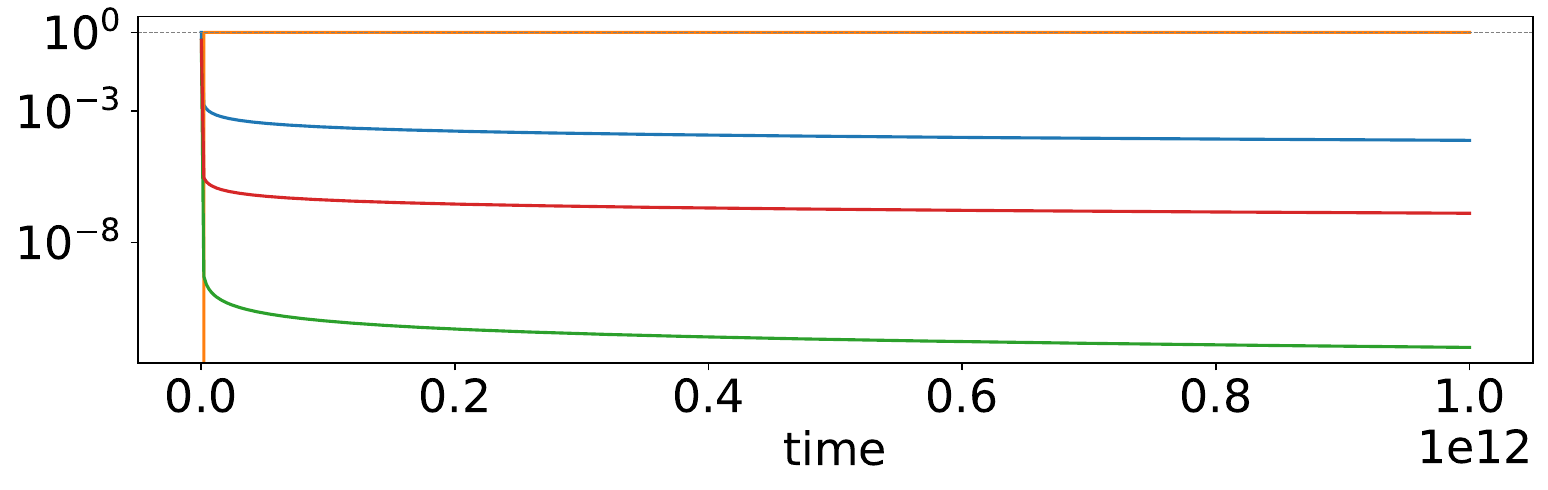}
    \caption{``Very long'' time, log $y$-scale.}
    \label{fig:majority_equalAB_very-long_unequal-rates}
\end{subfigure}
\caption{
    The majority CRN of \Cref{thm:majority}, similar to \Cref{fig:majority_lessA}, but starting with equal $A,B$: $A=B=0.5, C=0.5, Y=1$. 
    Since $A \not > B$, 
    $N$ should converge to 1 and $Y$ to 0.
    We show only ``adversarial'' rate constants as in
    \Cref{fig:majority_lessA_short_unequal-rates,fig:majority_lessA_long_unequal-rates}.
    Also $A(t)=B(t)$ for all $t$, so we omit a plot of $B$.
    This case demonstrates the need for $C$ to ``tie-break'': 
    both $A$ and $B$ converge to 0,
    so we rely on reaction \eqref{rxn:maj-cy} to convert all $Y$ to $N$.
    That in turn requires $C$ to converge to 0 asymptotically more slowly than $A$ to outcompete reaction \eqref{rxn:maj-an}.
    We focus on convergence to 0 by setting a logarithmic $y$-scale for \Cref{fig:majority_equalAB_medium_unequal-rates,fig:majority_equalAB_long_unequal-rates,fig:majority_equalAB_very-long_unequal-rates}.
    We start $C(0)=0.5$ instead of 1 to emphasize the asymptotics: 
    $C$ initially goes lower than $A$
    (\Cref{fig:majority_equalAB_short_unequal-rates}), 
    but eventually $A$ catches up and is much lower than $C$
    (\Cref{fig:majority_equalAB_medium_unequal-rates,fig:majority_equalAB_long_unequal-rates,fig:majority_equalAB_very-long_unequal-rates}).
    All rate constants are ``adversarial'' as in \Cref{fig:majority_lessA_short_unequal-rates,fig:majority_lessA_long_unequal-rates},
    10 for reactions \eqref{rxn:maj-an} and \eqref{rxn:maj-3c} and 1 for the others.
    Note that convergence takes much longer when $A=B$.
    Since $Y(t)+N(t)=1$ for all $t$, we can verify that $N$ approaches 1 because $Y$ approaches 0
    (decreasing final values of $Y$ in 
    \Cref{fig:majority_equalAB_short_unequal-rates,fig:majority_equalAB_medium_unequal-rates,fig:majority_equalAB_long_unequal-rates,fig:majority_equalAB_very-long_unequal-rates}), 
    although $Y$ converges to 0 even more slowly than $A$ or $C$.
}
\label{fig:majority_equalAB}
\end{figure}

\subsection{Piecewise affine function}
We give an example of a threshold-piecewise rational floor-affine function $f$ that is defined as:

\[
f = \begin{cases}
    x_1 - x_2 & \text{ if $x_1>x_2$,}
    \\
    \min(x_1, x_2) & \text{ if $x_2>x_1$.}
    \end{cases}
\]

CRC $\calC_1$ with input species $X_1^1,X_2^1$ (copies of input species $X_1,X_2$ to the main CRC) and output Species $Y_1$ with the following reactions computes $f_1 = x_1 - x_2$
\[
X_1^1 \rxn^{k_1} Y_1
\]
\[
X_2^1 + Y_1 \rxn^{k_2} \emptyset
\]

Similarly CRC $\calC_2$ with input species $X_1^2,X_2^2$ and output species $Y_2$ with the following reaction computes $f_2 = \min(x_1, x_2)$
\[
X_1^2 + X_2^2 \rxn^{k_3} Y_2
\]

We apply the construction from \cref{thm:majority,thm:piecewise-affine-are-robustly-computable} on $\calC_1,\calC_2$.
See \cref{fig:function} for plots of the resulting CRC that computes $f$.

\begin{figure}[htbp]
\centering
\includegraphics[width=0.6\textwidth]{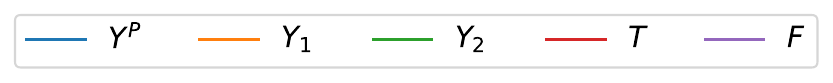}
\\
\begin{subfigure}[b]{0.49\textwidth}
    \centering
    \includegraphics[width=\textwidth]{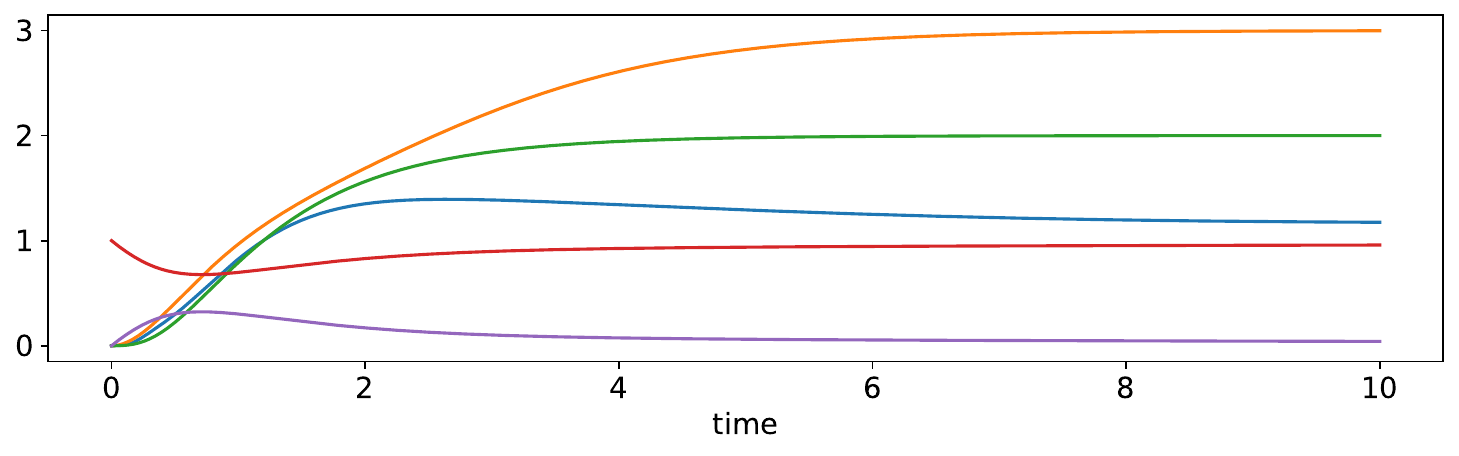}
    \caption{``Short'' time, underlying functions $f_1,f_2$ (output $f_j(\vx)=Y_j$ represented as $Y_j^P - Y_j^C + \hat Y_j^P - \hat Y_J^C$) and predicate have converged.}
    \label{fig:function_short}
\end{subfigure}
\hfill
\begin{subfigure}[b]{0.49\textwidth}
    \centering
    \includegraphics[width=\textwidth]{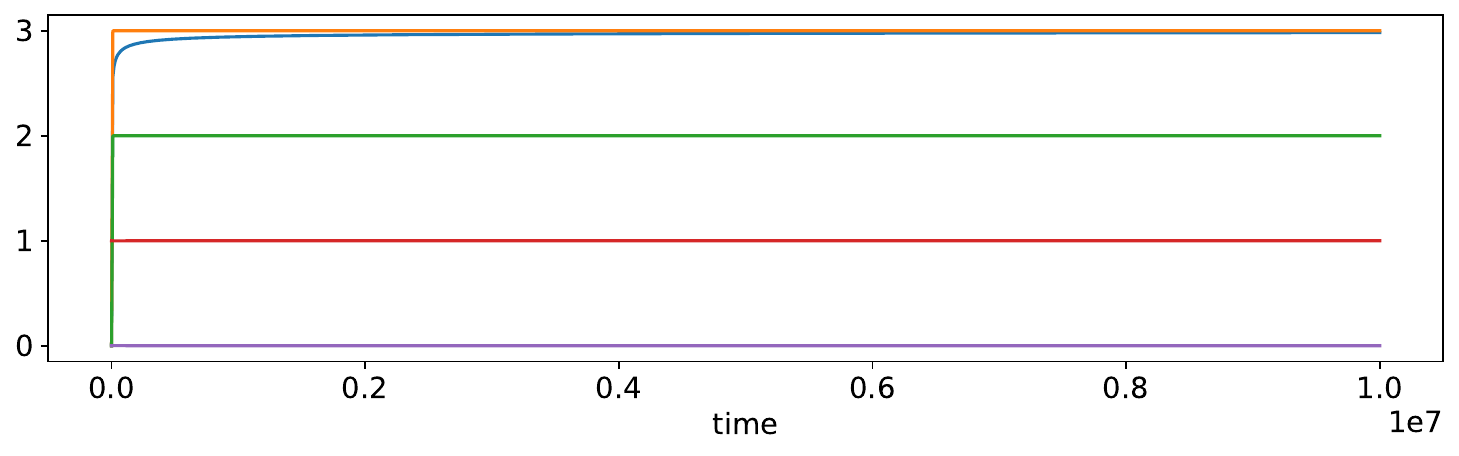}
    \caption{``Long'' time.
    Piecewise affine function (output species $Y^P$) 
    $f(x_1,x_2)$ ($= f_1(x_1,x_2)$ since $x_1 > x_2$)
    has converged.}
    \label{fig:function_long}
\end{subfigure}
\caption{
    Example illustrating construction of 
    \Cref{thm:piecewise-affine-are-robustly-computable},
    where $f_1(x_1,x_2) = x_1-x_2$ and $f_2(x_1,x_2) = \min(x_1,x_2)$,
    and $\phi_1(x_1,x_2) = 1 \iff x_1 > x_2 \iff \overline{\phi}_1(x_1,x_2)$,
    and $f(x_1,x_2) = f_1(x_1,x_2)$ if $\phi_1(x_1,x_2)=1$ (so $\phi_1(x_1,x_2)=0$)
    and  $f(x_1,x_2) = f_2(x_1,x_2)$ otherwise,
    starting with concentrations
    $X_1=5, X_2=2, C=1, T=1$.
    Since there are only two predicates here, which are negations of each other, for simplicity we simply have a CRD for only $\phi_1$ and let its yes and no voters $T,F$ play the reversed roles that the voters for $\phi_2$ would.
    Since $x_1 > x_2$,
    the correct output $f(x_1,x_2)$ is $f_1(x_1,x_2) = x_1-x_2 = 5-2=3.$
    \Cref{fig:function_short} shows that after a fairly short time,
    $f_1,f_2,\phi$ have converged.
    In particular, the proof of \Cref{thm:piecewise-affine-are-robustly-computable} shows that for $j\in\{1,2\}$, $Y_j^P - Y_j^C + \hat{Y}_j^P - \hat{Y}_j^C$ converges to $f_j$,
    so we plot those values to represent the output of the two CRCs computing $f_1$ and $f_2$ respectively, 
    i.e., what each of their outputs would be if simulated without the reactions of \Cref{thm:piecewise-affine-are-robustly-computable} 
    that ``activate'' those outputs converting them to the global output $Y$.
    Convergence of $\phi$ is shown by $T$ converging to 1 and $F$ to 0.
}
\label{fig:function}
\end{figure}

%% file: stable.tex
\section{Stable predicate computation by continuous CRNs}
\label{sec:stable_predicate_computation}
In this section, we show that the class of predicates stably decidable by continuous CRNs is exactly the detection predicates.
Intuitively, $\phi : \R_{\ge 0}^k \to \{0,1\}$ is a \emph{detection predicate} if it is a Boolean combination of questions in the form ``is initial concentration of species $S$ positive or not?''. We formalize this notion as follows.

\begin{definition}
\label{def:detection-predicate}
    A predicate $\phi:\Rp^k \to \{0,1\}$ is a \emph{simple detection predicate} if there is a $1 \leq i \leq k$ such that the extension of $\phi$ (the set $\phi^{-1}(1)$ of all input vectors that make $\phi$ true) is of the form 
    $
    \{\vx \in \Rp^k \mid \vx(i) > 0\}.
    $
    A \emph{detection predicate} is one expressible as a finite combination of ANDs, ORs, and NOTs of simple detection predicates.
\end{definition}

In other words, a simple detection predicate is defined by a hyperplane (with rational slopes),
with output 1 on one side of the hyperplane and output 0 on the other side and on the hyperplane itself.
A detection predicate $\phi$ is defined by a finite number of hyperplanes that partition $\Rp^k$ into a finite number of regions,
and $\phi$ is constant within each region.

The following is main result of this section.
Each direction is proven separately in \Cref{sec:stable-positive-result,sec:stable-negative-result} via 
\Cref{lem:stable-computation-computes-all-detection-predicates,lem:stable-computation-limited-to-detection}.

\begin{theorem}
\label{thm:stable-detection-predicates}
    $\phi: \Rp^k \to \{0,1\}$ is stably decidable by a continuous CRN if and only if $\phi$ is a detection predicate.
\end{theorem}

\subsection{Stable (rate-independent) computation}
\label{sec:prelim:stable}

These definitions are taken from~\cite{chen2023rate};
see Section 2.4 of that paper for justification that the notion of segment-reachability in particular
(\Cref{defn:reachable-lines})
is reasonable.

\begin{definition}\label{defn-reachable-line}
state $\vd$ is \emph{straight-line reachable (aka $1$-segment reachable)} from state $\vc$, written $\vc \slto \vd$, if $(\exists \vu \in \Rp^R)\ \vc + \vM \vu = \vd$
and $\vu(\alpha) > 0$ only if reaction $\alpha$ is applicable at $\vc$.
In this case write $\vc \slto_\vu \vd$.
\end{definition}

\noindent Intuitively, by a single segment we mean running the reactions applicable at $\vc$ at a constant (possibly 0) rate to get from $\vc$ to $\vd$.
In the definition, $\vu(\alpha)$ represents the flux, or total amount executed, of reaction $\alpha \in R$.

\begin{definition}
\label{defn:reachable-lines}
Let $k \in \N$.
state $\vd$ is \emph{$k$-segment reachable} from state $\vc$, written $\vc \segto^k \vd$, 
if $(\exists \vb_0, \dots, \vb_{k})\ \vc  = \vb_0 \slto \vb_1 \slto \vb_2 \slto \dots  \slto \vb_{k}$,
with $\vb_k = \vd$.
\end{definition}

\begin{definition}
\label{defn:reachable-segment}
state $\vd$ is \emph{segment-reachable} 
(or simply \emph{reachable}) 
from state $\vc$, written $\vc \segto \vd$, if $(\exists k\in\N)\ \vc \segto^k \vd$.
\end{definition}

\noindent Often \cref{defn:reachable-segment} is used implicitly, when we make statements such as, ``Run reaction 1 until $X$ is gone, then run reaction 2 until $Y$ is gone'', which implicitly defines two straight lines in concentration space.

We now formalize what it means for a CRN to ``rate-independently'' compute (stably decide) a predicate $\phi$.


We define the \emph{global output partial function} $\Phi : \N^\Lambda \dashrightarrow \{0,1\}$ as $\Phi(\vx) = 1$. If the only voter species with positive concentration in the state $\vx$ are yes voters, then $\Phi(\vx) = 1$. Conversely, if the only voter species with positive concentration are no voters, $\Phi(\vx) = 0$. Lastly, if neither of these conditions is met, the output function $\Phi(\vx)$ is undefined.
We say a state $\vo$ is \emph{stable} if, for all $\vc$ such that $\vo \segto \vc$: $\Phi(\vo) = \Phi(\vc)$.

\begin{definition}[stably decide]
\label{def:stably-decide}
    Let $\phi : \R^k_{\ge0} \to \{0,1\}$ be a predicate. We say a CRD $\calD$ \emph{stably decides} $\phi$ if, for all $\vx \in \R^k_{\ge0}$, and all $\vc$ such that $\vx \segto \vc$, there exists a stable state $\vo$ such that $\vc \segto \vo$ and $\Phi(\vo) = \phi(\vx)$. We say a set $A$ is \emph{stably decidable} if its indicator function $\chi_A$ is stably decidable.
\end{definition}

\begin{definition}[stably compute]
\label{def:stably-compute}
    Let $f:\domR^k \to \R$ be a function. 
    We say a CRC $\calC = (\Lambda,R,\Sigma,\{Y\},\sigma)$ \emph{stably computes} $f$ if for all $\vx \in \R^k_{\ge0}$ and all $\vc$ such that $\vx \segto \vc$, there exists a stable state $\vo$ such that $\vc \segto \vo$ and $\vo(Y) = f(\vx)$.
\end{definition}

\begin{theorem}
\label{thm:stable-computable-iff-piecewise-linear-positive-continuous}
    A function $f: \Rp^k \to \R$ is stably computable by a continuous CRC if and only if it is positive-continuous and piecewise rational linear.
\end{theorem}

A linear function's graph defines a $k$-dimensional hyperplane in $\Rp^{k+1}$ that passes through the origin (i.e., linear but not affine).
If a piecewise rational linear function is positive-continuous,
we switch from one linear component $f_i$ to another $f_j$ only where their hyperplanes intersect.
Thus the question ``\emph{is linear component $f_i$ the correct linear component to apply to compute $f$ on input $\vx$?}'' is itself a  \emph{multi-threshold predicates} as defined in \Cref{def:threshold-predicate}, but with constant $h=0$.

\subsection{Positive result: All detection predicates are stably decidable by continuous CRNs}
\label{sec:stable-positive-result}

To show that continuous CRNs can stability decide all detection predicates, we will utilize the fact that detection predicates are Boolean combinations of simple detection predicates. We first show that each simple detection predicate is decidable, then show that stable decidability is closed under the basic Boolean operations. 

\begin{lemma}
\label{lem:stable-computation-computes-all-detection-predicates}
    Every simple detection predicate is stably decidable by a continuous CRN. 
\end{lemma}
\begin{proof}
    To decide the simple detection predicate $\phi_i : \Rp^k \to \{0,1\}$ defined by $\phi_i(\vx) = 1$ if $\vx(i) > 0$ and $\phi(\vx) = 0$ otherwise, we let $\Lambda = \Sigma =\{X_1,\ldots,X_k\}$ be the set of species, and let the set of yes voter $\yesVotes = \{X_i\}$ and the set of no voters be each other species. For some $k \ne i$, the CRN will start with the initial context $\{1 X_k$\}.
    The reactions in the CRN are $X_i + X_j \to X_i$ for all $j \ne i$. 
    This stably decides $\phi_i$, as if any $X_i$ is present in a initial state then all species are eventually converted into copies of $X_i$. If $X_i$ is not present, then no reactions will be applicable, so none of the present no voters are converted.
\end{proof}

To see that all detection predicates can be computed, it suffices to show that stably-decidable predicates are closed under Boolean operations, which is shown in \Cref{lem:stable-decidability-closure-properties} below.
This is done by recalling a general method from the discrete model of CRNs~\cite[Lemma 6]{AngluinADFP2006}, which also works in the real-valued CRN model.

\begin{lemma}
\label{lem:stable-decidability-closure-properties}
    Let $\calD_1 = (\Lambda_1,R_1,\Sigma_1,\yesVotes^1,\noVotes^1,\sigma_1)$ and $\calD_2 = (\Lambda_2,R_2,\Sigma_2,\yesVotes^2,\noVotes^2,\sigma_2)$ be CRDs stably deciding predicates $\phi: \Rp^k \to \{0,1\}$ and $\psi: \Rp^k \to \{0,1\}$ respectively. Then there are CRDs stably deciding the predicates $\phi \lor \psi$, $\phi\land \psi$, and $\overline{\phi}$.
\end{lemma}

\begin{proof}

Given a CRD $\calD_1 = (\Lambda_1,R_1,\Sigma_1,\yesVotes^1,\noVotes^1,\sigma_1)$ that stably decides the predicate $\phi$, we can decide the predicate $\overline{\phi}$ by ``flipping'' all of our yes and no voters. That is, $\overline{\phi}$ is decided by the CRN $\calD_1 = (\Lambda_1,R_1,\Sigma_1,\noVotes^1,\yesVotes^1,\sigma_1)$. To decide $\phi \lor \psi$, we will run both $\calD_1$ and $\calD_2$ in parallel on the same input, and have their respective voters influence the concentration of a global voter species. That is, we construct the CRD $\calD$ as follows:
\begin{enumerate}
    \item Initialize the set of species as $\Lambda_1 \cup\Lambda_2$, the set of reactions as $R_1 \cup R_2$ and the set of input species as $\Sigma_1 \cup \Sigma_2$.
    \item For each input species $A \in \Sigma_1 \cup\Sigma_2$, replace each instance of $A$ in $\calC_1$ with species $A_1$ and each instance of $A$ in $\calC_2$ with the species $A_2$. Add the reaction $A \to A_1 + A_2$. Conceptually, this is copying the global input to the inputs of both $\calD_1$ and $\calD_2$.
    \item Add the species $\vnn$, $\vyn$, $\vny$ and $\vyy$. These species ``record'' the current vote of $\calD_1$ and $\calD_2$: The first subscript represents the current vote of $\calD_1$, and the second represents the current vote of $\calD_2$. For each yes voter $Y_1 \in \yesVotes^1$ and no voter $N_1 \in \noVotes^1$, add the reactions 
    \begin{align}
        \label{eq:yesvoters1}
    V_{\text{n}x} + Y_1 &\rxn V_{\text{y}x}+Y_1\\
    \label{eq:novoters1}
    V_{\text{y}x} + N_1 &\rxn V_{\text{n}x}+N_1
    \end{align}
    for each ${x} \in \{\text{y},\text{n}\}$. Similarly for each yes voter $Y_2$ and no voter $N_2$ in $\calD_2$ add the reactions
    \begin{align}
    \label{eq:yesvoters2}
    V_{x\text{n}} + Y_2 &\rxn V_{x\text{y}}+Y\\
    \label{eq:novoters2}
    V_{x\text{y}} + N_2 &\rxn V_{x\text{n}}+N_2
    \end{align}
    for each ${x} \in \{\text{y},\text{n}\}$. 
\item Let the yes voters be $\vyn$, $\vny$ and $\vyy$ and let the initial context be $\sigma_1 + \sigma_2 + \{1\vyy\}$
\end{enumerate}

To see that $\calD$ stably decides $\phi \lor \psi$, let $\vx$ be the input. For concreteness, say $\phi(\vx) = 1$ and $\psi(\vx) = 0$. The other cases are identical. Let $\vc$ be a state reachable from $\vx$ along with the initial context. We must show that a correct output state is reachable from $\vc$. First, if any of the global input species $A$ is present then run reaction $A \rxn A_1 + A_2$ until $A$ is gone. This gives a state $\vc'$ which we can partition as $\vc' = \vc_1 + \vc_2 + \vd$ where for each $i \in \{1,2\}$, $\vc_i$ contains species from $\calD_i$ (including the copied input species $A_i$) and $\vd$ contains the added voting species. We note that the reactions \eqref{eq:yesvoters1}, \eqref{eq:novoters1}, \eqref{eq:yesvoters2} and \eqref{eq:novoters2} do not change the concentration of species from $\calD_1$ or $\calD_2$, and reactions from $\calD_1$ do not affect species in $\calC_2$ (and vice versa). Therefore, $\vc_1$ is a state reachable in $\calD_1$ from the initial state $\vx + \sigma_1$. Since $\calC_1$ stably computes $\phi$, there is a stable state\footnote{Which is reachable in finitely many line segments by Theorem 2.14. in \cite{chen2023rate}} $\vo_1$ such that $\vc_1 \segto \vo_1$. Repeating this argument with $\vc_2$, we obtain stable $\vo_2$ such that $\vc_2 \segto \vo_2$. By additivity, $\vc_1 + \vc_2 + \vd \segto \vo_1 + \vo_2 + \vd$. This is a state where the only former voting species present are $Y_1$ and $N_2$. Therefore, we may apply the reaction \eqref{eq:yesvoters1} until all $\vnn$ species are gone. This stably decides $\phi \lor \psi$. To decide $\phi \land \psi$, we use the exact same CRN, except the only yes voter is $\vyy$. A similar argument shows the correctness of this CRN as well.



\end{proof}

\subsection{Negative result: All stably decidable predicates are detection predicates}
\label{sec:stable-negative-result}




In this section we show that, unlike the case of computing functions $f:\R^k \to \R$ with numerical output (the focus of~\cite{chen2023rate}),
CRNs stably computing \emph{predicates} are much more severely limited in computational power.
We prove this using the fact, proven in~\cite{chen2023rate}, that such functions $f$ must be
\emph{positive-continuous}, meaning that discontinuities can only occur when some input coordinate $\vx(i)$ goes from 0 to positive.
(Note the conceptual similarity to detection predicates, which can change their output only when some input coordinate $\vx(i)$ goes from 0 to positive.)
We connect this to predicates by showing that any CRN $\calC$ stably computing a predicate $\phi$ can be augmented to stably compute a function $\sigma_\phi$ that is continuous only if the predicate stably computed by $\calC$ is a detection predicate, proving that since such functions $\sigma_\phi$ must be continuous, then $\phi$ must be a detection predicate.

We recall the definition of a positive-continuous function from \cite{chen2023rate}.
Intuitively, a positive-continuous function is only allowed to have discontinuities whenever some input goes from 0 to positive.
For example, the function $f(x_1,x_2) = x_1$ if $x_2 > 0$ and $f(x_1,x_2)=0$ otherwise is positive continuous though not continuous.
\begin{definition}
\label{defn:positive-continious}
    A function $f : \Rp^k \to \Rp$ is \emph{positive-continuous}, for all $U \subseteq\{1,\ldots,k\}$, $f$ is continuous on the domain
    \[
    D_U = \left\{\vx \in \R^k_{\ge0} \mid \vx(i) > 0 \iff i \in U\right\}.
    \]
\end{definition}
To understand the definition, it helps to understand first what the domains $D_U$ look like. In the case $k = 2$. The domains $D_U$ are the origin, the positive $x$ and $y$ axes, and the set $\{(x,y) \mid x,y > 0\}$. If $U = \{1\}$, then $D_U$ is the positive $x$ axis as the vectors $\vx \in D_U$ must satisfy $\vx(1) >0$ and $\vx(2) = 0$. A positive continuous function is one which must be continuous inside each domain, but is allowed discontinuities as it moves across boundaries. If useful, a positive continuous function can be understood as a collection of $2^k$ continuous functions, each defined on a different piece of the positive orthant.
\todo{KC: Add picture illustrating domains in $\R^2$}

Given a predicate $\phi : \R_{\ge 0}^k \to \{0,1\}$, define the \emph{sum characteristic function} of $\phi$,  
$\sigma_\phi: \R_{\ge 0}^k \to \R_{\ge 0}$ for all $\vx \in \Rp^k$ by
\[
\sigma_\phi(\vx) = \begin{cases}
    \sum_{i=1}^k \vx(i) = \| \vx \|_1 & \text{if $\phi(\vx) = 1$}\\
    0 & \text{if $\phi(\vx) = 0$}
\end{cases}
\]
 
\begin{lemma}
\label{lem:non-detection-pred-implies-non-positive-continous-char-func}
    Let $\phi : \R_{\ge 0}^k \to \{0,1\}$.
    If $\phi$'s sum characteristic function $\sigma_\phi$ is positive-continuous, then $\phi$ is a detection predicate.
\end{lemma}

\begin{proof}
We show the contrapositive, that if $\phi$ is not a detection predicate, then the induced $\sigma_\phi$ is not positive-continuous. 
Suppose $\phi$ is not a detection predicate. 
Then, for some $U \subseteq \{1,\ldots,k\}$ the region $D_U$ contains points $\vx,\vy \in D_U$ such that $\phi(\vx) \ne \phi(\vy)$, as all detection predicates cannot change their output within a domain $D_U$.
Note that this implies $D_U \ne \{\vec{0}\}$,
i.e., $U \neq \emptyset$. 
Let $\ell$ denote the straight line connecting $\vx$ and $\vy$ in $\R^k$. Such an $\ell$ is completely contained in $D_U$, as $D_U$ is a convex set. 
Consider the image of $\ell$ under $\sigma_\phi$, denoted $\sigma_\phi(\ell) \subseteq \R_{\geq 0}$.
Since at least one point along $\ell$ does not satisfy $\phi$, 
we have $0 \in \sigma_\phi(\ell)$. 
Furthermore, since $\vx,\vy \neq \vec{0}$, there is $\varepsilon > 0$ such that $\| \vx \|_1 > \varepsilon$ and $\| \vy \|_1 > \varepsilon$.
Since $\| \cdot \|_1$ is a linear function, this implies that for all points $\vz \in \ell$, $\| \vz \|_1 > \varepsilon$.
Thus, all non-zero elements of $\sigma_\phi(\ell)$ are greater than $\varepsilon$.
This implies $0$ is an isolated point of $\sigma_\phi(\ell)$. 
Since $\sigma_\phi(\ell)$ contains an isolated point, it is not a connected set.  
As a continuous function must preserve connectedness, this implies that  $\sigma_\phi$ is not continuous on $D_U$, i.e. it is not positive continuous.
\end{proof}

\begin{lemma}
\label{lem:stable-computation-limited-to-detection}
    If $\phi: \Rp^k \to \{0,1\}$ is stably decidable by a CRD,
    then $\phi$ is a detection predicate.
\end{lemma}



\begin{proof}
    Let $\calD$ be a CRD stably deciding $\phi$.
    We show how to convert $\calD$ into a CRC $\calC$ that stably computes the sum characteristic function $\sigma_\phi$.
    Since all functions stably computable by a CRC are positive-continuous~\cite{chen2023rate}, \Cref{lem:non-detection-pred-implies-non-positive-continous-char-func} implies that $\phi$ must be a detection predicate.

    For each input species $X_i$,
    add the reaction $X_i \rxn X_i' + Y$, where $X_i'$ is the $i$'th input species to $\calD$.
    For each yes-voter $T$ in $\calD$, add the reaction
    $T + \hat Y \rxn T + Y$,
    and for each no-voter $F$,
    add the reaction $F + Y \rxn F + \hat Y$.

    Let $\vc$ be a state reached from the initial state $\vx$. If any of the input species $X_i$ are present in $\vc$, we apply the reaction $X_i \rxn X_i' + Y$ until all these species are gone from the state. Once this is done,
    the concentrations of $Y$ and $\hat{Y}$ satisfy $[Y] + [\hat Y] = \sum_{i=1}^k \vx_i.$ Denote this reached state by $\vc'$.
    Since $\calD$ stably decides $\phi$, there is a stable state $\vo$ reachable from $\vc'$ in which all species present are yes-voters if $\phi(\vx) = 1$ and no-voters otherwise. 
    In the former case,
    we may apply the reaction $T + \hat Y \rxn T + Y$ to convert all $\hat Y$ to $Y$. As no voters are present in the state, it is stable. Furthermore, the concentration of $\hat{Y}$ is exactly $\left\|\vx\right\|_1$.
    If $\phi(\vx) = 0$, then only no voters are present so running the reaction $F + Y \rxn F + \hat Y$ eventually will remove all $Y$, stabilizing on the correct output of $0$.
\end{proof}




%% file: conclusion.tex
\section{Conclusion}
\label{sec:conclusion}

Motivated by the limitations of stable predicate computation, we investigated the robust computation of predicates and numerical functions.
While we established positive results on what can be robustly computed,
namely multi-threshold predicates
(\Cref{def:threshold-predicate})
and robustly piecewise floor-affine functions
(\Cref{def:piecewise-affine}),
the limitations of robust computation remains an open question.
We conjecture the positive results are tight, i.e., \emph{exactly} the multi-threshold predicates and robustly piecewise floor-affine functions are robustly computable.

We have assumed the presence of initial context, for example to help include a positive amount of some voter species in all of our constructions,
which gave the nice invariant that the sum of voter concentrations is preserved.
This choice seems as though it is not strictly necessary, since our construction could instead generate voter species.
We conjecture that a leaderless model, without initial context,
would only restrict threshold predicates to a constant threshold of 0 and constrain functions to piecewise rational floor-\emph{linear} functions, rather than affine.

We assumed that the output species $Y$ 
\todo{DD: oops, we should have removed this paragraph because we didn't need that assumption after all after Kim's fix}
approaches static equilibrium in the definition of stable function computation,
but we conjecture that this necessarily holds for any CRC that robustly computes a function.
This is known (and much easier to see) to hold for stable computation,
since stable computation requires the much stronger condition that we reach a ``stable'' state from which \emph{no possible} reactions could change $Y$, so clearly all reactions changing $Y$ must be disabled in a stable state.

Another possible extension of this work is employing the so-called \emph{dual-rail} representation~\cite{chen2023rate} to accommodate negative values as inputs.
This means that a (possibly negative) value $x$ is represented as the difference between two nonnegative species concentrations $X^+ - X^-$.
This is known to slightly reduce the class of stably computable functions $f:\R^k \to \R$;
such a function is stably computable using the dual-rail representation (for both inputs and output)
if and only if it is piecewise rational linear and \emph{continuous}.

In the construction given in 
\cref{thm:piecewise-affine-are-robustly-computable}, we employed a technique to modify upstream CRCs, ensuring controlled composition even when certain reactions within these CRCs consumed the output species. This structured approach to dependency management could naturally extend to the broader tricky problem of CRN composition~\cite{chalk2019composable,severson2021composable,hashemi2020composable}.